\documentclass[journal]{IEEEtran}

\usepackage{amssymb}
\usepackage[cmex10]{amsmath}
\usepackage{stfloats}
\usepackage{graphicx}
\usepackage{subfigure}
\usepackage{tabularx}
\usepackage{epsfig,epsf,color,balance,cite}
\usepackage{verbatim}
\usepackage{url}
\usepackage{bm}
\usepackage{epstopdf}
\usepackage{enumerate}

\newtheorem{lemma}{Lemma}
\newtheorem{proof}{Proof}[section]
\usepackage{algorithm}
\usepackage{algpseudocode}
\usepackage{amsmath}
\usepackage{lipsum}
\usepackage{multicol}
\usepackage{amsmath,epsfig,amssymb,subfigure,bm,dsfont}
\usepackage{cuted}

\begin{document}
	
	\title{{Reflections in the Sky: Joint Trajectory and Passive Beamforming Design for Secure UAV Networks with Reconfigurable Intelligent Surface }}
	
\author{
	\IEEEauthorblockN{Hui~Long$^{1}$,
		Ming Chen$^{1}$,
		Zhaohui Yang$^{2}$,
		Bao Wang$^1$,Zhiyang~Li$^1$,
		 Xu Yun$^1$, and Mohammad Shikh-Bahaei$^{2}$}\\
	\IEEEauthorblockA{$^1$National Mobile Communications Research Laboratory,
		Southeast University, Nanjing 210096, China.\\
		$^2$Centre for Telecommunications Research, Department of Informatics, King's College London, London WC2B 4BG, UK\\
		Emails: longhui@seu.edu.cn, chenming@seu.edu.cn, yang.zhaohui@kcl.a.uk, 220180884@seu.edu.cn,lizhiyang@seu.edu.cn,  xuyun@seu.edu.cn and m.sbahaei@kcl.ac.uk
}}
	
	\maketitle
	\pagestyle{empty}
	\thispagestyle{empty}
	\begin{abstract}
		This paper investigates a problem of secure energy efficiency maximization for a reconfigurable intelligent surface (RIS) assisted uplink wireless communication system, where an unmanned aerial vehicle (UAV) equipped with an RIS works as a mobile relay between the base station (BS) and a group of users. We focus on maximizing the secure energy efficiency of the system via jointly optimizing the UAV's trajectory, the RIS's phase shift, user association and transmit power. To tackle this problem, we divide the original problem into three sub-problems, and propose an efficient iterative algorithm. In particular, the successive convex approximation (SCA) method is applied to solve the nonconvex UAV trajectory, the RIS's phase shift, and transmit power optimization sub-problems. We further provide two schemes to simplify the solution of phase and trajectory sub-problem. Simulation results demonstrate that the proposed algorithm converges fast, and the proposed design can enhance the secure energy efficiency by up to 38\% gains, as compared to the traditional schemes without any RIS.
	\end{abstract}
	

	\IEEEpeerreviewmaketitle
	\vspace{-0.75em}
	\section{Introduction}
	Unmanned aerial vehicle (UAV) is playing an increasingly important role in line of sight (LoS) instant communication\cite{x2019physical}. Typical applications of UAV-assisted communication have established the safe LoS communication links with ground nodes utilizing its flexible networking structure and feasible low deployment cost \cite{Y2016throughout,7063641,j2017placement,Caching7875131,yang8379427,zhou8434285,wang2019uav,yang8764580}. Achieving secure transmission of confidential information and avoiding eavesdropping have always been an important problem in the design of wireless communication systems. In the existing research on physical layer security, the communication nodes are static which means the channel quality of eavesdropper-base station (BS) or legal user-BS mainly depends on the location of eavesdropper and legal users. If the distance between the BS and the legitimate receiver/edvesdropper is fixed, the achievable secrecy rate will be limited even if techniques such as artificial noise (AN) and power control are applied. However, UAV not only establishes stronger legitimate links with the ground nodes by designing its trajectory, but also can detect any potential eavesdroppers by equipping with optical cameras. As a result, in this paper, we focus on the application of UAV on the physical layer security.
	
   A number of existing works such as \cite{G2017securing,li2019joint,8525328} have studied on security for various UAV communication
   systems. Considering a typical three-node eavesdropping scenario, the work in \cite{G2017securing} improved the average secrecy rate through optimizing the UAV trajectory under total transmit power constraint of the BS. Considering the general multi-user scenario, the authors in \cite{li2019joint} maximized the minimum secrecy rate by controlling the user association as well as considering the trajectory and transmit power. Different from the works on the ground to air communication system in 2-dimensional (2D) space \cite{G2017securing}\cite{li2019joint}, the work in \cite{8525328} investigated the more complicated air-to-air (A2A) systems in 3-dimensional (3D) space. By characterizing the statistical characteristics of the signal-to-noise-ratio (SNR) over the A2A links, the authors in \cite{8525328} obtained the closed-form expressions for secrecy outage probability.

    Recently, reconfigurable intelligent surface (RIS) has been widely used in improving the energy efficiency and communication equality in wireless networks \cite{li2020reconfigurable,huang8741198,xu9024490,yang2020energyefficient,yang2020RSMARIS,huang2019holographic,pan9090356,pan2019intelligent}. An RIS is made up of a number of configurable elements that can reflect the incident signal by controlling its phase shifts appropriately. Different from the traditional amplify and forward (AF) relays, RIS is almost passive, and does not incur energy cost\cite{e2019woreless}. Besides, an RIS will not induce or amplify noise while reflecting signals. It is worth noticing that there have been several works on RIS taking the secrecy into consideration including multiple-input single-output (MISO) systems\cite{8743496,8847342,9014322,8723525} and multiple-input multiple-output (MIMO) systems \cite{8972406}. The authors in \cite{8743496} maximized the power of received signal under the transmit power and the unit modulus constraints, and further improved secrecy rate for a  MISO system. For the MIMO case, the authors in \cite{8972406} proposed an iterative optimization method to maximize the secrecy rate with respect to the RIS's phase shift coefficient and the transmit covariance. However, the aforementioned works in \cite{8743496,8847342,9014322,8723525,8972406} do not exploit the RIS's characteristic in the mobile UAV-enabled communication system even though UAV can further improve the secrecy rate performance. Meanwhile, energy efficiency has emerged as an important performance index for deploying green and sustainable wireless networks \cite{Chen2019A},\cite{8846706}. In consequence, it is of importance to investigate the application of RIS in improving secrecy energy efficiency, which is defined as the ratio of the minimum secrecy rate to the total power consumption.

    Motivated by the previous works, we investigate the secrecy energy efficiency of an UAV-enabled system by taking the advantages of RIS in this paper.  Considering the complex outdoor environment, tall building may block the LoS communication links between ground users and BS which seriously affects the channel quality. To improve the channel quality between users and the BS, an UAV relay with one RIS is considered in this paper. Our goal is to maximize the secrecy energy efficiency. The main contributions of this paper are summarized as follows:
     \begin{itemize}
     	\item  To solve this nonconvex secrecy energy efficiency maximization problem, an alternating method is proposed with solving three sub-problems iteratively. For the integer user association sub-problem, it is relaxed to a linear problem. For the power control sub-problem, the successive convex approximation (SCA) method is adopted. For the phase and trajectory optimization sub-problem, the optimal phase to maximize the user rate is derived first, and then an alternating method is proposed by analyzing the convexity of the secrecy rate expression with respect to the trajectory variable.
     	\item To compare the method for the phase and trajectory optimization sub-problem, we also provide the detailed procedures to jointly optimize phase and trajectory  via SCA.
     	\item  Simulation results show that our proposed approach can enhance the secrecy energy efficiency by up to 38\% gains compared to the conventional AF relay scheme.
     \end{itemize}
The rest of this paper is organized as follows. System model and problem formulation are described in Section \uppercase\expandafter{\romannumeral2}. Section \uppercase\expandafter{\romannumeral3} provides the  algorithm design and analysis. Section \uppercase\expandafter{\romannumeral4} presents the simulation results to demonstrate the performance of the proposed algorithm. Conclusions are drawn in Section \uppercase\expandafter{\romannumeral5}.

	\section{System Model and Problem Formulation}
	
	\subsection{System Model}
	\begin{figure}[t]
		\centering
		\includegraphics[width=3.5in]{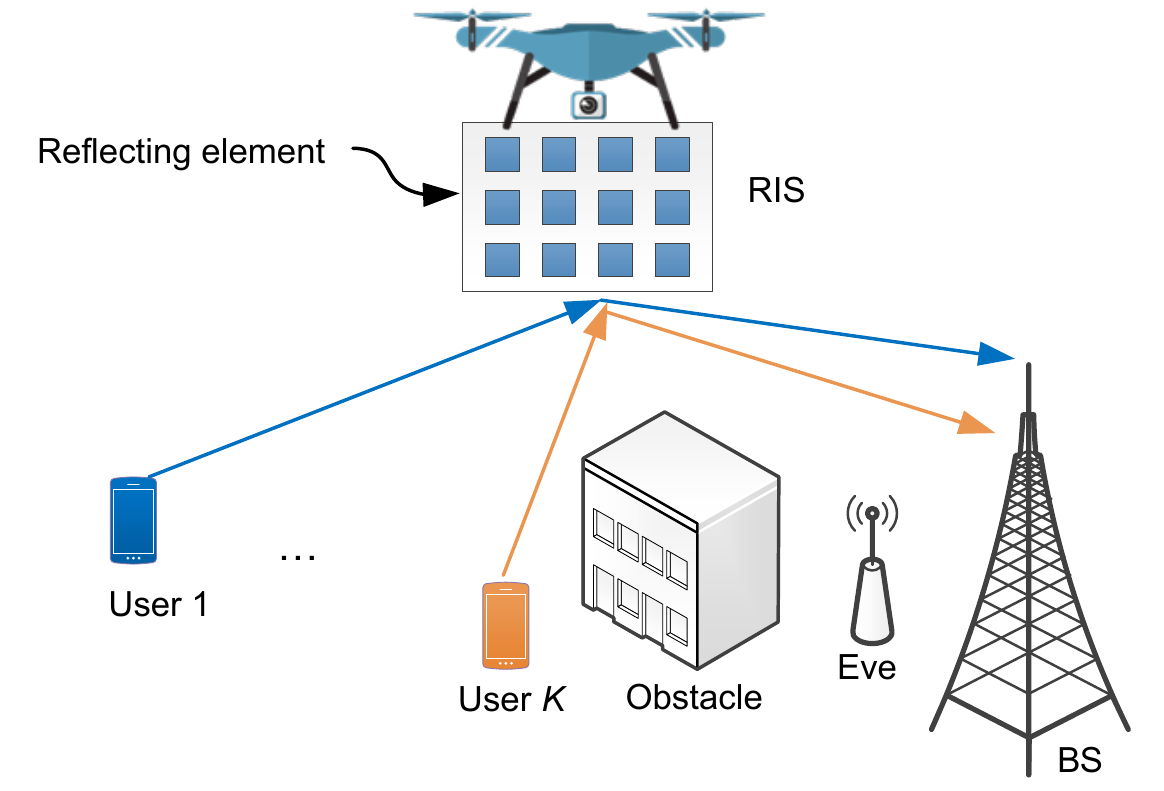}
		\caption{An uplink RIS-assisted wireless communication system.} \label{figsys1}
	\end{figure}
	
	Consider an uplink wireless communication system with one UAV, one eavesdropper (Eve), $K$ users, and one BS, as shown in Fig. \ref{figsys1}.
	The set of $K$ users is denoted by $\mathcal K =\{1, \cdots, K\}$.
	Due to the obstacle of high walls, there is non LoS channel between the BS and each user.
	The UAV equipped with one RIS serves as a passive relay to assist the communication between the users and the BS.
	The RIS is equipped with a uniform linear array (ULA) of $M$ reflecting elements and the phase of each element can be controlled by the UAV.
	
	Eve, all users, and the BS are located at the ground.
	The horizontal coordinates of user $k$ where $k \in \mathcal K$ , Eve, and BS are denoted by $\boldsymbol w_k=[x_k, y_k]^T$, $\boldsymbol w_{e}=[x_{e}, y_{e}]^T$, and $\boldsymbol w_{b}=[x_{b}, y_{b}]^T$, respectively.
	In this system, multiple users are served in different time intervals.
	The UAV flies at a fixed altitude $H$ with flight period $T$.
	To facilitate the analysis,
	the UAV flight period $T$ is divided into $N$ equally-spaced time slots with step size $\delta$, i.e., $T=N\delta$.
	Denote $\mathcal N=\{1, \cdots, N\}$ as the set of all discrete time slots.
	The time-variant horizontal coordinate of the UAV in time slot $n$ is denoted by $\boldsymbol q[n]=[x[n],y[n]]^T$, $n\in\mathcal N$.
	
	To serve the users periodically, the UAV needs to return back to the initial position by the end of period $T$, i.e.,
	\begin{equation}\label{sys1eq1}
	\boldsymbol q[N]=\boldsymbol q[0],
	\end{equation}
	where $\boldsymbol q[0]=[x[0],y[0]]^T$ is the predetermined initial horizontal coordinate of the UAV.
	With given maximal UAV speed $V_{\max}$, the number of time slots $N$ can be chosen properly such that the time for UAV location changing within time $\delta$ can be negligible.
	As a result, we have
	\begin{equation}\label{sys1eq2}
	\|\boldsymbol q[n]-\boldsymbol q[{n-1}]\| \leq S_{\max},
	\end{equation}
	where $S_{\max}=V_{\max}\delta$ is the maximal horizontal distance that the UAV can travel within one time slot.
	
	The channel gain between user $k$ and the UAV in time slot $n$ can be expressed as \cite{li2020reconfigurable}
	\begin{equation}\label{sys1eq3_1}
	\boldsymbol g_k[n]=\sqrt{h_0 d_{k}^{-\alpha}[n]}[1, e^{-j\frac{2\pi d}{\lambda}\phi_k[n]},\cdots,e^{-j\frac{2(M-1)\pi d}{\lambda}\phi_k[n]}]^T,
	\end{equation}
	where $h_0$ is the channel gain at a reference distance $d_0=1$ m, $d_{k}[n]=\sqrt{\|\boldsymbol q[n]-\boldsymbol  w_k\|^2+H^2}$ is the  distance between user $k$ and the UAV in time slot $n$, $\alpha\geq2$ is the pathloss exponent,
	$\phi_k[n]=\frac{x[n]-x_k}{d_{k}[n]}$
	represents the cosine of the angle of
	arrival (AoA) of the signal from user $k$ to the ULA at the
	RIS in  time slot $n$,
	$d$ is the antenna separation, and $\lambda$ is
	the carrier wavelength.
	
	Similarly, the channel gain between the UAV and the BS in time slot $n$ can is given by
	\begin{equation}\label{sys1eq3_2}
	\boldsymbol g_b[n]=\sqrt{h_0 d_{b}^{-\alpha}[n]}[1, e^{-j\frac{2\pi d}{\lambda}\phi_b[n]},\cdots,e^{-j\frac{2(M-1)\pi d}{\lambda}\phi_b[n]}]^T,
	\end{equation}
	where $d_{b}[n]=\sqrt{\|\boldsymbol q[n]-\boldsymbol  w_b\|^2+H^2}$ and
	$\phi_b[n]=\frac{x[n]-x_b}{d_{b}[n]}$.
	The channel gain between the UAV and Eve in time slot $n$ can be expressed as
	\begin{equation}\label{sys1eq3_3}
	\boldsymbol g_e[n]=\sqrt{h_0 d_{e}^{-\alpha}[n]}[1, e^{-j\frac{2\pi d}{\lambda}\phi_e[n]},\cdots,e^{-j\frac{2(M-1)\pi d}{\lambda}\phi_e[n]}]^T,
	\end{equation}
	where $d_{e}[n]=\sqrt{\|\boldsymbol q[n]-\boldsymbol  w_e\|^2+H^2}$ and
	$\phi_e[n]=\frac{x[n]-x_e}{d_{e}[n]}$.
	
	Let the binary variable $a_k[n]$ denote the association of user $k$ in time slot $n$, i.e., $a_k[n]=1$ represents that user $k$ is associated with the UAV; otherwise $a_k[n]=0$.
	Assume that at most one user is served in each time slot, i.e.,
	\begin{equation}
	\sum_{k=1}^K a_k[n] \leq 1, \quad \forall n\in\mathcal N
	\end{equation}
	Based on \eqref{sys1eq3_1} and \eqref{sys1eq3_2}, the achievable rate from user $k$ to the BS via RIS in time slot $n$ can be given by
	\begin{equation}\label{sys1eq5}
	r_k[n]= \log_2 \left(1+\frac{ p_k[n] | \boldsymbol g_{b}^H[n]\boldsymbol\Theta[n] \boldsymbol g_k[n]|^2}{\sigma^2}\right),
	\end{equation}
	where
	$\boldsymbol \Theta[n]$ is the phase shift matrix of the RIS in time slot $n$,
	$\sigma^2$ is the noise power, and $p_k[n]$ is the transmit power of user $k$ in time slot $n$. Matrix $\boldsymbol \Theta[n]=\text{diag} (\text e^{j\theta_1[n]}, \cdots, \text e^{j\theta_M[n]})\in\mathbb C^{M\times  M}$  with $\theta_{m}[n]\in[0,2\pi]$, which captures the effective phase shifts applied by all reflecting elements of the RIS.
	If user $k$ transmits data to the BS in time slot $n$, the achievable rate at Eve is
	\begin{equation}\label{sys1eq6}
	c_k[n]= \log_2 \left(1+\frac{ p_k[n] | \boldsymbol g_{e}^H[n]\boldsymbol\Theta[n] \boldsymbol g_k[n]|^2}{\sigma^2}\right).
	\end{equation}
	
	According to \cite{li2019joint}, the secrecy rate of user $k$ in time slot $n$ can be expressed as
	\begin{equation}\label{sys1eq6}
	R_k[n]=a_k[n]\max\{r_k[n]-c_k[n],0\}.
	\end{equation}
	Considering the fairness among all users, we provide the minimum secrecy rate of all users as follow
	\begin{equation}\label{sys1eq6}
	\zeta=\min_{k\in\mathcal K} \frac 1 N \sum_{n=1}^NR_k[n].
	\end{equation}
	
	\subsection{Problem Formulation}
	Our aim is to maximize the fair secrecy energy efficiency, via joint trajectory, transmit power, and passive beamforming optimization.
	Mathematically, the optimization problem can be formulated as
	\begin{subequations}\label{max1}
		\begin{align}
		\!\!\mathop{\max}_{\zeta,\boldsymbol{Q}, \boldsymbol A, \boldsymbol{P}, \boldsymbol\Theta}\quad
		& \frac{\zeta}{\sum_{k=1}^K \sum_{n=1}^N p_k[n]+P_0}\\
		\textrm{s.t.}\quad\:\:
		& \zeta \leq \frac 1 N \sum_{n=1}^NR_k[n], \quad \forall k\in\mathcal K\\
		&\|\boldsymbol q[n]-\boldsymbol q[{n-1}]\| \leq S_{\max}, \quad \forall n\in\mathcal N\\
		&\sum_{k=1}^K a_k[n] \leq 1, \quad \forall n\in\mathcal N\label{max1_aa}\\
		&a_{k}[n]\in\{0,1\},  \quad \forall n\in\mathcal N,k\in\mathcal K\label{max1_a}\\
		&p_k[n] \leq P_k, \quad \forall n\in\mathcal N,k\in\mathcal K\\
		& \boldsymbol q[N]=\boldsymbol q[0],\\
		&\theta_{m}[n] \in[0,2\pi], \quad \forall m \in\mathcal M,n\in\mathcal N,
		\end{align}
	\end{subequations}
	where $\boldsymbol Q=\{\boldsymbol q[n]\}_{\forall n}$,
	$\boldsymbol A=\{ a_k[n]\}_{\forall k,n}$,
	$\boldsymbol P=\{ p_k[n]\}_{\forall k,n}$,
	$\boldsymbol \Theta=\{ \theta_m[n]\}_{\forall m,n}$,
	$P_0$ is the circuit power consumption of the system,
	$P_k$ is the maximum transmit power of user $k$,
	and $\mathcal M=\{1, \cdots, M\}_{\forall m \in\mathcal M}$.

    \section{Proposed Algorithm}
	 Problem \eqref{max1} is a nonconvex problem due to the nonconvex and discrete constraints. In this section, a sub-optimal solution that contains the SCA and alternating methods is provided to tackle with  problem \eqref{max1}.
	\subsection{User Association Optimization}
    To make problem \eqref{max1} tractable, we relax the binary variables in \eqref{max1_a} into continuous variables. Thus, with given transmit power $\boldsymbol P $, RIS's phase shift matrix $\boldsymbol \Theta$, and UAV trajectory $\boldsymbol Q$, the user association problem can be optimized by solving the following problem
    \begin{subequations}\label{maxa}
		\begin{align}
		\!\!\mathop{\max}_{\zeta, \boldsymbol A}\quad
		& \zeta\\
		\textrm{s.t.}\quad\:\:
		& \zeta \leq \frac 1 N \sum_{n=1}^NR_k[n], \quad \forall k\in\mathcal K\\
		&\sum_{k=1}^K a_k[n] \leq 1, \quad \forall n\in\mathcal N\\
		&0\leq a_{k}[n]\leq 1,  \quad \forall n\in\mathcal N,k\in\mathcal K.\label{maxa1}\end{align}
	\end{subequations}
    Problem \eqref{maxa} is a standard linear programming and can be solved efficiently by existing optimization tools such as CVX. We can obtain the optimal solution of problem \eqref{maxa}, and then using the rounding method to further get the integer solution.

    \subsection{Power Optimization}
    With fixed user association $\boldsymbol A $, UAV trajectory $\boldsymbol Q $, and RIS's phase shift matrix $\boldsymbol \Theta$, the transmit power optimization problem reduces to
    \begin{subequations}\label{maxp}
		\begin{align}
		\!\!\mathop{\max}_{\zeta, \boldsymbol{P}}\quad
		& \frac{\zeta}{\sum_{k=1}^K \sum_{n=1}^N p_k[n]+P_0}\\
		\textrm{s.t.}\quad\:\:
		& \zeta \leq \frac 1 N \sum_{n=1}^NR_k[n], \quad \forall k\in\mathcal K\label{maxp1}\\
		&p_k[n] \leq P_k, \quad \forall n\in\mathcal N,k\in\mathcal K.\label{maxp2}\end{align}
	\end{subequations}
     Problem \eqref{maxp} is nonconvex as $R_k[n]$ in the right hand side of constraint \eqref{maxp1} is a
     difference of two concave functions with respect to the power
     control variables $p_k[n]$. We can use the SCA method to solve problem \eqref{maxp} by sequentially solving a series of convex approximation problems. Let $p_k^{(r)}[n]$ denote the transmit power of user $k$ in the $r$-th iteration of SCA method. As the $c_k[n]$ is the concave function of $p_k[n]$, we have
     \begin{align}\label{sca}
     \notag
     c_k[n]&= \log_2 \left(1+\frac{ p_k[n] | \boldsymbol g_{e}^H[n]\boldsymbol\Theta[n] \boldsymbol g_k[n]|^2}{\sigma^2}\right)\\
     \notag
     &\leq \log_2 \left(1+\frac{ p_k^{(r)}[n] | \boldsymbol g_{e}^H[n]\boldsymbol\Theta[n] \boldsymbol g_k[n]|^2}{\sigma^2}\right)\\
     \notag
     &+\frac{| \boldsymbol g_{e}^H[n]\boldsymbol\Theta[n] \boldsymbol g_k[n]|^2\left(p_k[n]-p_k^{(r)}[n]\right)}{\ln2\left(\sigma^2+p_k^{(r)}[n]| \boldsymbol g_{e}^H[n]\boldsymbol\Theta[n] \boldsymbol g_k[n]|^2\right)}\\
     &\triangleq\hat{c}_k[n].
	 \end{align}
     Replacing $c_k[n]$ with $\hat{c}_k[n]$, problem \eqref{maxp} is equivalent to
\begin{subequations}\label{maxpnew1}
	 \begin{align}
		\!\!\mathop{\max}_{\zeta, \boldsymbol{P}}\quad
		& \frac{\zeta}{\sum_{k=1}^K \sum_{n=1}^N p_k[n]+P_0}\\
		\textrm{s.t.}\quad\:\:
		& \zeta \leq \frac 1 N \sum_{n=1}^Na_k[n]\max\{r_k[n]-\hat{c}_k[n],0\}, \quad \forall k\in\mathcal K\label{maxpnew1-1}\\
		&p_k[n] \leq P_k, \quad \forall n\in\mathcal N,k\in\mathcal K.\label{maxpnew1-2}\end{align}
\end{subequations}
    Problem \eqref{maxpnew1} is a nonlinear fractional programming\cite{dinkelbach1967on}. We solve this problem by adding a multiplication factor $\eta$, and problem \eqref{maxpnew1} is approximated as the following problem
\begin{subequations}\label{maxpnew2}
	 \begin{align}
		\!\!\mathop{\max}_{\zeta, \boldsymbol{P}}\quad
		& \zeta-\eta\left(\sum_{k=1}^K \sum_{n=1}^N p_k[n]+P_0\right)\\
		\textrm{s.t.}\quad\:\:
		& \zeta \leq \frac 1 N \sum_{n=1}^Na_k[n]\max\{r_k[n]-\hat{c}_k[n],0\}, \forall k\in\mathcal K\label{maxpnew2-1}\\
		&p_k[n] \leq P_k, \quad \forall n\in\mathcal N,k\in\mathcal K\label{maxpnew2-2},\end{align}
\end{subequations}
    where $\eta$ can be obtained by Algorithm 1. It is proved that the solution of problem \eqref{maxpnew1} can be obtained by solving problem \eqref{maxpnew2} according to reference \cite{dinkelbach1967on}.
  \begin{algorithm}[!t]
  	\caption{Iterative Algorithm for Problem \eqref{maxp} }
  	\label{alg:1}
  	\begin{algorithmic}[1]
  		\State \textbf{Initialize} $\zeta^{(0)}$, $\boldsymbol P^{(0)}$, the tolerance $\varepsilon$, and the iteration number $t=1$.
  		\State \textbf{repeat}
  		\State Calculate
  		\vspace{-.5em}
  		\begin{equation}\vspace{-.5em}
  		\eta^{(t)}=\frac{\zeta^{(t-1)}}{\sum_{k=1}^K \sum_{n=1}^N p_k^{(t-1)}[n]+P_0}.
  		\end{equation}
  		\State Solve the following optimization problem
  		\begin{small}
  			\begin{align}\label{maxpend}
  			\notag
  			&\left({{\zeta^{(t)}},{\boldsymbol{P}^{(t)}}}\right)=\mathop{\arg\max}_{\zeta, \boldsymbol{P}}\quad\zeta-\eta^{(t)}\left(\sum_{k=1}^K \sum_{n=1}^N p_k[n]+P_0\right)\\
  			\notag
  			\textrm{s.t.}\:\:
  			& \zeta \leq \frac 1 N \sum_{n=1}^N\max a_k[n]\{r_k[n]-\hat{c}_k[n],0\}, \quad \forall k\in\mathcal K\\
  			&p_k[n] \leq P_k, \quad \forall n\in\mathcal N,k\in\mathcal K.\end{align}
  		\end{small}
  		\State Update $t = t+1$.
  		\State \textbf{until} $\left|\zeta^{(t)}-\zeta^{(t-1)}\right|\leq\varepsilon$.
  	\end{algorithmic}
  \end{algorithm}

The suboptimal solution of problem \eqref{maxp} can be obtained through the iterative algorithm in Algorithm 1. Note that problem \eqref{maxpend} in Algorithm 1 is a convex problem, which can be easily solved by the standard toolbox, such as CVX.

     \subsection{Joint Phase and UAV Trajectory Optimization (Scheme I)}
     For any given user association $\boldsymbol A $ and transmit power $\boldsymbol P $, the optimization problem of UAV trajectory $\boldsymbol{Q}$ and RIS's phase shift matrix $\boldsymbol \Theta$ can be reformulated as
     \begin{subequations}\label{maxq}
		\begin{align}
		\!\!\mathop{\max}_{\zeta,\boldsymbol{Q},\boldsymbol \Theta}\quad
		& \zeta\\
		\textrm{s.t.}\quad\:\:
		& \zeta \leq \frac 1 N \sum_{n=1}^NR_k[n], \quad \forall k\in\mathcal K\\
		&\|\boldsymbol q[n]-\boldsymbol q[{n-1}]\| \leq S_{\max}, \quad \forall n\in\mathcal N\\
		& \boldsymbol q[N]=\boldsymbol q[0]\\
        &\theta_{m}[n] \in[0,2\pi],   \quad  \forall   m \in\mathcal M,n\in\mathcal N.
		\end{align}
	\end{subequations}
   Problem \eqref{maxq} can be solved in two steps: passive beamforming optimization and UAV trajectory optimization. 
   \subsubsection{Passive Beamforming Optimization}
   For the optimization of $\boldsymbol \Theta$, we can align the phases of the received signal  at BS to maximize the received signal energy. 

    Firstly, considering the optimization of $\boldsymbol \Theta\ $ with any given $\boldsymbol{Q}$, expression $\boldsymbol g_{b}^H[n]\boldsymbol\Theta[n] \boldsymbol g_k[n]$ can be written as
    \begin{equation}\label{theta1}
    \boldsymbol g_{b}^H[n]\boldsymbol\Theta[n] \boldsymbol g_k[n]=\frac{h_0\sum\limits_{m=1}^Me^{j\left(\theta_m[n]+\frac{2(m-1)\pi d}{\lambda}\left(\phi_b[n]-\phi_k[n]\right)\right)}}{\sqrt{d^{\alpha}_b[n]d^{\alpha}_k[n]}}.
    \end{equation}
    To maximize the received signal energy, we combine the signals from different paths coherently at BS. Thus, we set $\theta_1[n]=\theta_2[n]+\frac{2\pi d}{\lambda}\left(\phi_b[n]-\phi_k[n]\right)=...=\theta_M[n]+\frac{2\pi(M-1) d}{\lambda}\left(\phi_b[n]-\phi_k[n]\right)=\omega$, or re-expressed as
    \begin{equation}\label{theta6}
    \theta_m[n]=\frac{2\pi(m-1) d}{\lambda}\left(\phi_k[n]-\phi_b[n]\right)+\omega,\quad \forall m,n,k
    \end{equation}
    where $\omega \in [0,2\pi]$. In that way, we achieve the phase alignment of the received signal and further maximize the received signal energy. Thus, $\boldsymbol g_{b}^H[n]\boldsymbol\Theta[n] \boldsymbol g_k[n]$ can be rewritten as
    \begin{equation}\label{theta2}
    \boldsymbol g_{b}^H[n]\boldsymbol\Theta[n] \boldsymbol g_k[n]=\frac{h_0Me^{j\omega}}{\sqrt{d^{\alpha}_b[n]d^{\alpha}_k[n]}}.
    \end{equation}
    Based on \eqref{theta2}, the achievable rate of user $k$ in time slot $n$ can be rewritten as
    \begin{equation}\label{theta3}
    r_k[n]=\log_2\left(1+\frac{B}{d^{\alpha}_b[n]d^{\alpha}_k[n]}\right),
    \end{equation}
    where $B=\frac{p_k[n]\left| h_0\right| ^2M^2}{\sigma^2}$.

    Given the phase shift results in \eqref{theta6}, $\left|\boldsymbol g_{e}^H[n]\boldsymbol\Theta[n] \boldsymbol g_k[n]\right|^2$ can be written as
    \begin{equation}\label{theta4}
    \left|\boldsymbol g_{e}^H[n]\boldsymbol\Theta[n] \boldsymbol g_k[n]\right|^2
    =\frac{\left| h_0\right| ^2C^2}{d^{\alpha}_k[n]d^{\alpha}_e[n]},
    \end{equation}
    where $C = \left|\sum\limits_{m=1}^Me^{j\left(\theta_m[n]+\frac{2(m-1)\pi d}{\lambda}\left(\phi_e[n]-\phi_k[n]\right)\right)}\right|$.

    The effective channel gain between the user and Eve in \eqref{theta4} is a complicated function of the UAV trajectory. To make the problem tractable, we provide the following upper bound, i.e.,
    \begin{equation}\label{theta5}
    \left|\boldsymbol g_{e}^H[n]\boldsymbol\Theta[n] \boldsymbol g_k[n]\right|^2\leq\frac{\left| h_0\right| ^2M^2}{d^{\alpha}_k[n]d^{\alpha}_e[n]}.
    \end{equation}
\subsubsection{UAV Trajectory Optimization}
    Based on \eqref{theta3} and \eqref{theta5}, we first introduce slack variables $\boldsymbol{Z}=\left\{z_k[n]\right\}_{\forall k,n}$ and $\boldsymbol{V} =  \left\{v_k[n]\right\}_{\forall k,n} $. Problem \eqref{maxq} can be reformulated as
        \begin{subequations}\label{maxq1}
		\begin{align}
		\!\!\mathop{\max}_{\zeta,\boldsymbol{Q},\boldsymbol{Z},\boldsymbol{V}}\quad
		& \zeta\\
		\textrm{s.t.}\quad\:\:
		& \zeta \leq\frac 1 N \sum_{n=1}^Na_k[n]\max\left\{r_k[n]-c_k[n],0\right\},\forall k\in\mathcal K\label{maxq2}\\
		&\|\boldsymbol q[n]-\boldsymbol q[{n-1}]\| \leq S_{\max}, \quad \forall n\in\mathcal N\\
		& \boldsymbol q[N]=\boldsymbol q[0]\\
        &z^{\frac{2}{\alpha}}_k[n]\geq d_k^2[n]d_b^2[n], \quad \forall n\in\mathcal N,k\in\mathcal K\label{maxq3}\\
        &v^{\frac{2}{\alpha}}_k[n]\leq d_k^2[n]d_e^2[n], \quad \forall n\in\mathcal N,k\in\mathcal K,\label{maxq4}
		\end{align}
	\end{subequations}
where $r_k[n]\!=\!\log_2\left(1+\frac{B}{z_k[n]}\right)$ and  $c_k[n]\!=\!\log_2\left(1+\frac{B}{v_k[n]}\right)$. To maximize the  objective value $\zeta$, constraints \eqref{maxq3} and \eqref{maxq4} must hold with equality at the optimal solution of problem \eqref{maxq1}. To solve problem \eqref{maxq1}, we introduce an important lemma as follows.
    \begin{lemma}
     Given $B>0$, function $r_k[n]$ is convex  with respect to $z_k[n]>0$.
    \end{lemma}
    \begin{proof}
    See Appendix A.
   \end{proof}
   Using the same method in Appendix A, $c_k[n]$ can be proved to be convex with respect to $v_k[n]$.

   With Lemma 1, constraint \eqref{maxq2} is the difference of two convex functions. Set $z^{(l)}_k[n]$ as the given points $\boldsymbol{Z}=\left\{z_k[n]\right\}_{\forall k,n}$ in the $l$-th iteration, we obtain the following lower bound for $r_k[n]$ in \eqref{maxq2}, i.e., 
   \begin{align}\label{maxq5}
   \notag
   r_k[n]& =\log_2\left(1+\frac{B}{z_k[n]}\right)\\
   & \geq
   \notag \log_2\left(1+\frac{B}{z^{(l)}_k[n]}\right)\\
   &+\frac{-B\left(z_k[n]-z^{(l)}_k[n]\right)}{z^{(l)}_k[n](z^{(l)}_k[n]+A)\ln2}\triangleq\hat{r}_k[n].
   \end{align}
   Therefore, the nonconvexity of constraint \eqref{maxq2} can be handled based on \eqref{maxq5}.  Before tackling constraint \eqref{maxq3} and \eqref{maxq4}, we show another important lemma as follows.
   \begin{lemma}
   	$d_k^2[n]$, $d_k^4[n]$ are convex functions of $\boldsymbol{q}[n]$.
   \end{lemma}
   \begin{proof}
	See Appendix B.
\end{proof}
   Lemma 2 is proved by checking the Hessian matrix of $d_k^2[n]$, $d_k^4[n]$.
   Similarly, $d_b^2[n]$, $d_e^2[n]$, $d_b^4[n]$, $d_e^4[n]$ are all convex functions of $\boldsymbol{q}[n]$.

   Problem \eqref{maxq1} is still nonconvex due to the nonconvex terms $d_k^2[n]d_b^2[n]$ and $d_e^2[n]d_b^2[n]$. We can obtain the upper bound function of $d_k^2[n]d_b^2[n]$ via its first-order Taylor expansion at any given point $d^{(l)}_b[n]$, $d^{(l)}_k[n]$, $d^{(l)}_e[n]$ and $\boldsymbol{q}^{(l)}[n]$, i.e.,
   \begin{align}\label{maxq7}
   \notag
   &d_k^2[n]d_b^2[n]\\
   \notag
   &=\frac 1 2\left[\left(d_k^2[n]+d_b^2[n]\right)^2-\left(d_k^4[n]+d_b^4[n]\right)\right]\\
   \notag
   &\leq \frac 1 2\left[\left(d_k^2[n]+d_b^2[n]\right)^2-\left(\left(d_k^{(l)}[n]\right)^4+\left(d_b^{(l)}[n]\right)^4\right)\right]\\
   \notag
   &-2\left(d_k^{(l)}[n]\right)^2\left(\boldsymbol{q}^{(l)}[n]-\boldsymbol{w}_k\right)^T\cdot\left(\boldsymbol{q}[n]-\boldsymbol{q}^{(l)}[n]\right)\\
   \notag
   &-2\left(d_b^{(l)}[n]\right)^2\left(\boldsymbol{q}^{(l)}[n]-\boldsymbol{w}_b\right)^T\cdot\left(\boldsymbol{q}[n]-\boldsymbol{q}^{(l)}[n]\right)\\
   &\triangleq f\left(\boldsymbol{q}[n]\right).
   \end{align}
   Similarly, the lower bound function of  $d_k^2[n]d_e^2[n]$ can be written as
   \begin{align}\label{maxq6}
   \notag
   &d_k^2[n]d_e^2[n]\\
   \notag
   &=\frac 1 2\left[\left(d_k^2[n]+d_e^2[n]\right)^2-\left(d_k^4[n]+d_e^4[n]\right)\right]\\
   \notag
   &\geq \frac 1 2 \left[\left(\left(d_k^{(l)}[n]\right)^2+\left(d_b^{(l)}[n]\right)^2\right)^2-\left(d_k^4[n]+d_e^4[n]\right)\right]\\
   \notag
   &+2\left(d_k^{(l)}[n]\right)^2\left(2\boldsymbol{q}^{(l)}[n]-\boldsymbol{w}_k-\boldsymbol{w}_e\right)^T\cdot\left(\boldsymbol{q}[n]-\boldsymbol{q}^{(l)}[n]\right)\\
   \notag
   &+2\left(d_e^{(l)}[n]\right)^2\left(2\boldsymbol{q}^{(l)}[n]-\boldsymbol{w}_k-\boldsymbol{w}_e\right)^T\cdot\left(\boldsymbol{q}[n]-\boldsymbol{q}^{(l)}[n]\right)\\
   &\triangleq g\left(\boldsymbol{q}[n]\right).
   \end{align}

  Note that $v^{\frac{2}{\alpha}}_k[n]$ is concave, the first-order Taylor expansions of $v^{\frac{2}{\alpha}}_k[n]$ at the given point $v^{(l)}_k[n]$ is applied to make constraint \eqref{maxq4} feasible for convex optimization, thus
   \begin{align}\label{maxq9}
   v^{\frac{2}{\alpha}}_k[n]
  \leq
  \notag
  &\left(v^{(l)}_k[n]\right)^{\frac{2}{\alpha}}+\frac{2}{\alpha}\left(v^{(l)}_k[n]\right)^{\frac{2}{\alpha}-1}\left(v_k[n]-v^{(l)}_k[n]\right)\\
   \triangleq &h\left(v_k[n]\right).
   \end{align}

   With \eqref{maxq5}-\eqref{maxq9}, the optimization of UAV trajectory can be formulated as
        \begin{subequations}\label{maxq8}
		\begin{align}
		\!\!\mathop{\max}_{\zeta,\boldsymbol{Q},\boldsymbol{Z},\boldsymbol{V}}\quad
		& \zeta\\
		\textrm{s.t.}\quad\:\:
		& \zeta \leq \frac 1 N \sum_{n=1}^Na_k[n]\left[\hat{r}_k[n]-\log_2\left(1+\frac{B}{v_k[n]}\right)\right]^+\\
		&\|\boldsymbol q[n]-\boldsymbol q[{n-1}]\| \leq S_{\max}, \quad \forall n\in\mathcal N\\
		& \boldsymbol q[N]=\boldsymbol q[0]\\
        &z^{\frac{2}{\alpha}}_k[n]\geq f\left(\boldsymbol{q}[n]\right)\quad \forall n\in\mathcal N,k\in\mathcal K\\
        &h\left(v_k[n]\right)\leq g\left(\boldsymbol{q}[n]\right)\quad \forall n\in\mathcal N,k\in\mathcal K.
		\end{align}
	\end{subequations}
   Since all the constraints in \eqref{maxq8} are convex and the objective function is linear, problem \eqref{maxq8} is convex, which can be effectively solved via the conventional methods such as CVX. Note that problem \eqref{maxq} also can be solved by joint optimization with SCA method (scheme II) as shown in Appendix C.

   \subsection{Overall Algorithm and Analysis}

   In summary, the overall algorithm for solving problem \eqref{max1} is given in Algorithm 2.
   The convergence as well as complexity of the proposed Algorithm 2 are given as follows. 
    \begin{lemma}
   	Algorithm 2 is guaranteed to converge.
   \end{lemma}
   \begin{proof}
   	See Appendix D.
   \end{proof}
    
    From Algorithm 2, the complexity of solving problem \eqref{max1} is dominated by the complexity of solving three sub-problems : user association sub-problem \eqref{maxa},  power control sub-problem \eqref{maxp}, the phase and trajectory optimization sub-problem \eqref{maxq}. 
      
   The complexity of user association sub-problem is the  total number of variables $KN$, as problem \eqref{maxa} is a standard linear programing problem. The power control sub-problem is solved by SCA method and its complexity depends on its variables and constraints. Since there are $KN+1$ constraints in problem \eqref{maxp}, the number of iteration required for SCA method is $\mathcal{O}\left(  \sqrt{KN+1} \log_2{(1/\epsilon_1})\right) $, where $\epsilon_1$ is the accuracy of SCA method for solving problem \eqref{maxp}. Note that  $S_1 = KN+1$ is the total number of variables and $S_2 = K(N+1)$ is the total number of constraints. Thus, the complexity of solving problem \eqref{maxpnew2} at each iteration is $\mathcal{O} \left(S_1^2S_2\right)$. In consequence, the complexity of solving problem \eqref{maxp} is $\mathcal{O}\left((KN+1)^{2.5}(KN+K)\log_2{(1/\epsilon_1}) \right) $ ( For simplicity, it can be equivalent to $\mathcal{O}\left((KN)^{3.5}\log_2{(1/\epsilon_1}) \right) $. We only provide the simplidied forms in the following analysis.). In scheme I, the complexity for solving problem \eqref{maxq} is dominated by the  complexity of solving problem \eqref{maxq8} $\mathcal{O}\left((2KN)^{3.5}\log_2{(1/\epsilon_2}) \right) $, where $\epsilon_2$ is the accuracy of SCA method for solving problem \eqref{maxq8}. Similarly, in scheme II the complexity of solving problem \eqref{newmaxqall1_1} is $\mathcal{O}\left(((6K+M)N)^{3.5}\log_2{(1/\epsilon_3} \right) $,  where $\epsilon_3$ is the accuracy of SCA method for solving problem \eqref{newmaxqall1_1}.
   
   In conclusion, the total complexity of Algorithm 2 with scheme I for solving problem \eqref{max1} is \begin{small}$\mathcal{O} \left(S_3\left( KN+(KN)^{3.5}\log_2{(1/\epsilon_1})+(2KN)^{3.5}\log_2{(1/\epsilon_2})\right)  \right) $\end{small}, where $S_3$ is the number of iteration for Algorithm 2. It is obvious that the complexity of the proposed scheme I is lower than that of the proposed scheme II.
  \begin{algorithm}[!t]
  	\caption{Proposed Algorithm for Problem \eqref{max1} }
  	\label{alg:2}
  	\begin{algorithmic}[1]
  		\State \textbf{Initialize} $\boldsymbol P^{(0)}$, $\boldsymbol Q^{(0)}$, $\boldsymbol\Theta^{(0)}$, $\Gamma^{(0)}$, the tolerance $\varepsilon$, the iteration number $t=0$, the convergence accuracy $\epsilon$.
  		\State \textbf{repeat}
  		\State \quad With given $\left(\boldsymbol P^{(t)}, \boldsymbol Q^{(t)}, \boldsymbol\Theta^{(t)}\right)$, obtain $\boldsymbol A^{(t+1)}$ by solving problem \eqref{maxa}.
  		\State \quad With given $\left(\boldsymbol A^{(t+1)},\boldsymbol Q^{(t)},\boldsymbol\Theta^{(t)}\right)$,obtain $\boldsymbol P^{(t+1)}$ by Algorithm 1.
  		\State \quad With given $\left(\boldsymbol A^{(t+1)},\boldsymbol P^{(t+1)}\right)$, obtain $\boldsymbol Q^{(t+1)}$ by solving problem \eqref{maxq8}.
  		\State \quad With $\boldsymbol Q^{(t+1)} $, update $\boldsymbol\Theta^{(t+1)}$ by using \eqref{theta6}.
  		\State \quad Calculate $\Gamma^{(t+1)}=\frac{\zeta^{(t+1)}}{\sum_{k=1}^K \sum_{n=1}^N p^{(t+1)}_k[n]+P_0}$.
  		\State \quad If $\Gamma^{(t+1)}< \Gamma^{(t)}$, we set $Q^{(t+1)}=Q^{(t)}, \Theta^{(t+1)}=\Theta^{(t)}$.
  		\State \quad Set $t \leftarrow t+1$.
  		\State \textbf{Until} $\left|\Gamma^{(t+1)}-\Gamma^{(t)} \right|\leq\epsilon$.
  	\end{algorithmic}
  \end{algorithm}

\section{Numerical Results}

\begin{figure}[t]
	\centering
	\includegraphics[width=3.5in]{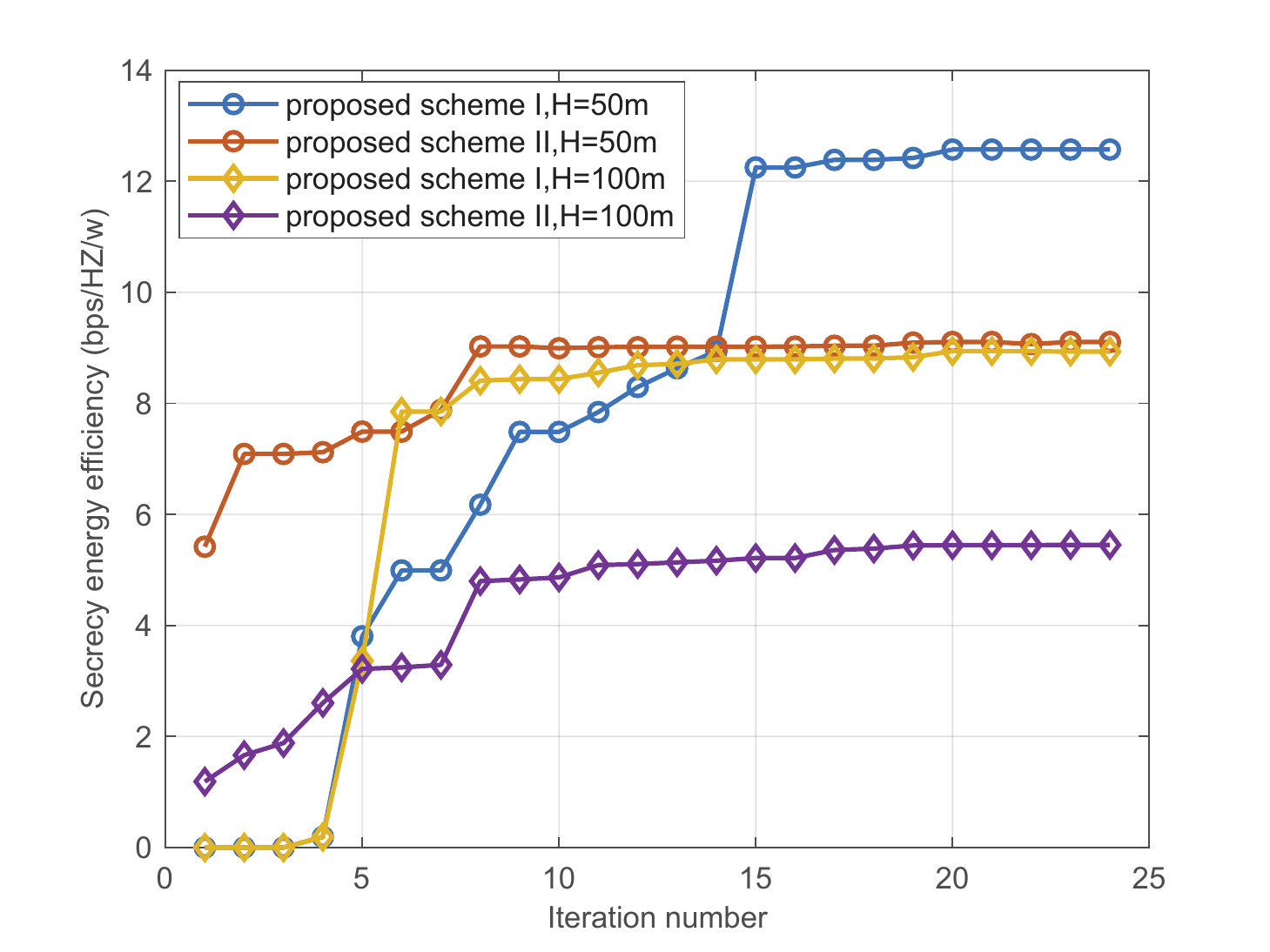}
	\caption{Convergence behavior with different
		optimization schemes.} \label{figcompare2}
	\vspace{-1em}
\end{figure}
In this section, we analyze the performance of the proposed algorithm through numerical results. Considering a square area of  300 m $\times$ 300 m with the BS located at the center. There are 4 users located at the four vertices of the square area with the coordinates $(\pm 300, \pm300)$ m, and the location of Eve is (0, 200) m. The flight period $T$ is fixed to 80 s, and the time slot is taken as $N$ = 12. We set the pathloss exponent $\alpha$ = 2.2,  the RIS's elements $M = 10$, $d/\lambda = 0.5$, the channel gain $h_0 = -80$ dB, noise power $\sigma^2$ = -120 dBm, and the system circuit power $P_0 = $ 1 W. Each user's  maximum transmit power $P_k$ is set to 1 W. The initial trajectory points of UAV flight trajectory are equally spaced on a polygon.

Fig. \ref{figcompare2} compare the convergence
performance between the proposed scheme I and scheme II  in different heights. As shown in Fig. \ref{figcompare2}, the performance gap between Scheme I and Scheme II shows the secrecy energy efficiency gain brought by the proposed scheme I. This is because that too many approximations are applied (e.g. the cosine of the AoA $\phi_k[n]$, $\phi_b[n]$, $\phi_e[n]$ in \eqref{newphik}-\eqref{newphie}, and the first order Taylor expansion in \eqref{eq48} and \eqref{eq49} in the proposed optimization scheme II.  Hence, we only consider the scheme I in the following simulation.

Fig. \ref{figN16} shows the UAV optimal trajectory solved by the proposed algorithm. In this simulation, the maximal UAV speed $V_{\max}$ is 50 m/s, which means that the maximal horizontal distance $S_{\max}$ in Fig. \ref{figN16} is 333.3 m. The UAV flight height $H$ is fixed at 100 m. The UAV flight trajectory can be optimized by Algorithm 2 to achieve the secrecy rate within the $S_{\max}$.
It is observed that UAV visits all user sequentially and its trajectory is a closed loop. On one hand, it is worth noting that UAV travels less distance when approaching user 4 and user 1, comparing with the initial trajectory. This shows that the UAV can establish communication with the users away from eavesdroppers Eve and achieve high secrecy rates. On the other hand, some points are marked with black in the optimal trajectory, which indicates that in these time slots there is no communication between any user and the BS. This is because UAV at these black points is close to the eavesdropper Eve which can not establish secure communication with users.

Fig. \ref{figN16v30} presents the trajectory when the maximal UAV speed $v_{\max}$ is fixed at 30 m/s, and the maximal horizontal distance $S_{\max}$ = 200 m. It can be found that the trajectory is smaller than the trajectory at $v_{\max} = 50$ m/s, which means with the decrease of the maximal UAV speed $v_{max}$, UAV will adjust it trajectory and fly closer to BS to achieve secure transmission of confidential information. Similar to the situation $v_{\max} = 50$ m/s, when UAV is close to the eavesdropper Eve which is marked with black points in Fig \ref{figN16v30}, UAV keeps silence to avoid information leakage.
\begin{figure}[t]
	\centering	
	\includegraphics[width=3.5in]{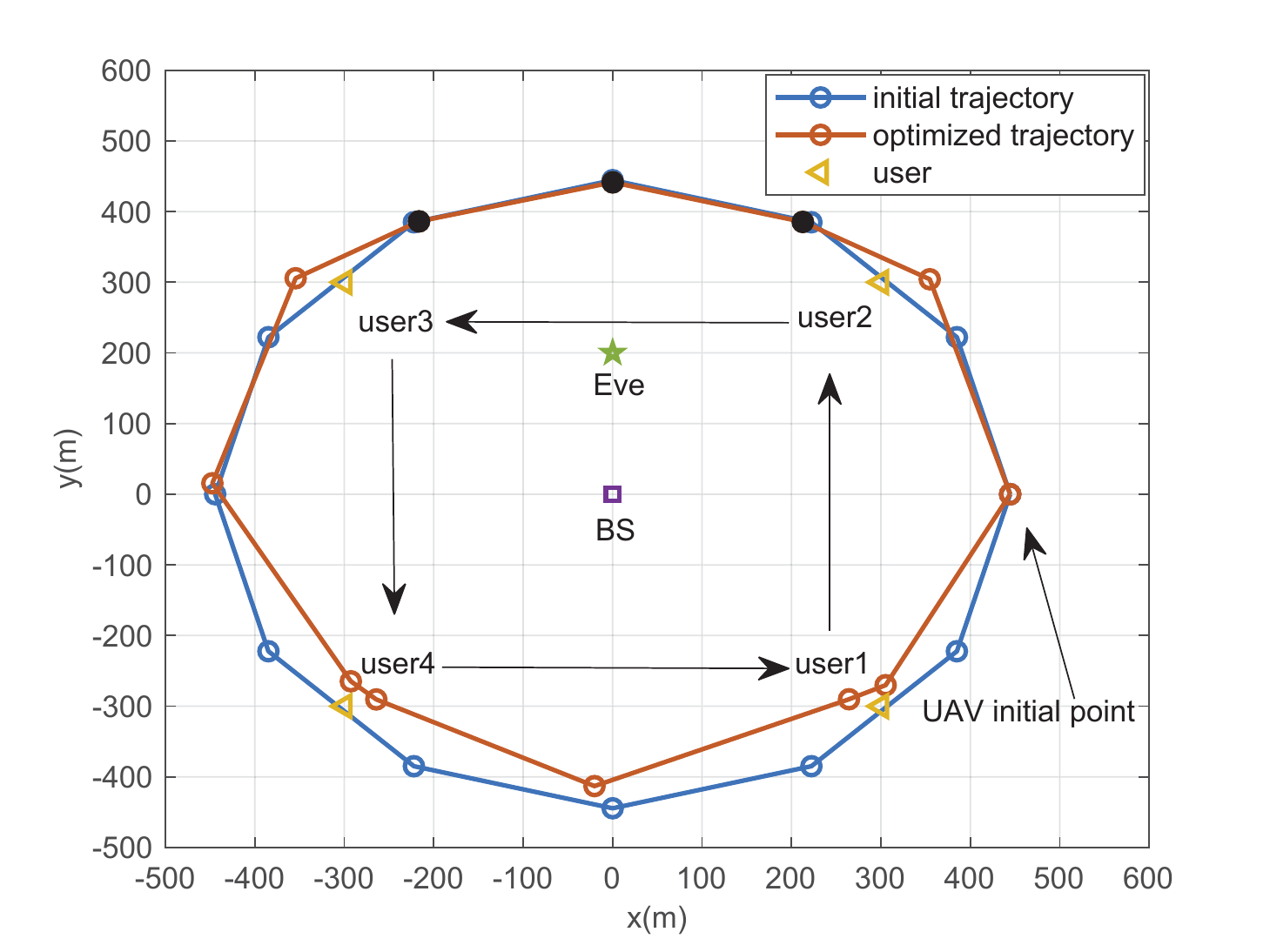}
	\caption{Optimized UAV trajectory by Algorithm 2 ($v_{\max} = 50$ m/s).} \label{figN16}
\end{figure}
\begin{figure}[t]
	\centering	
	\includegraphics[width=3.5in]{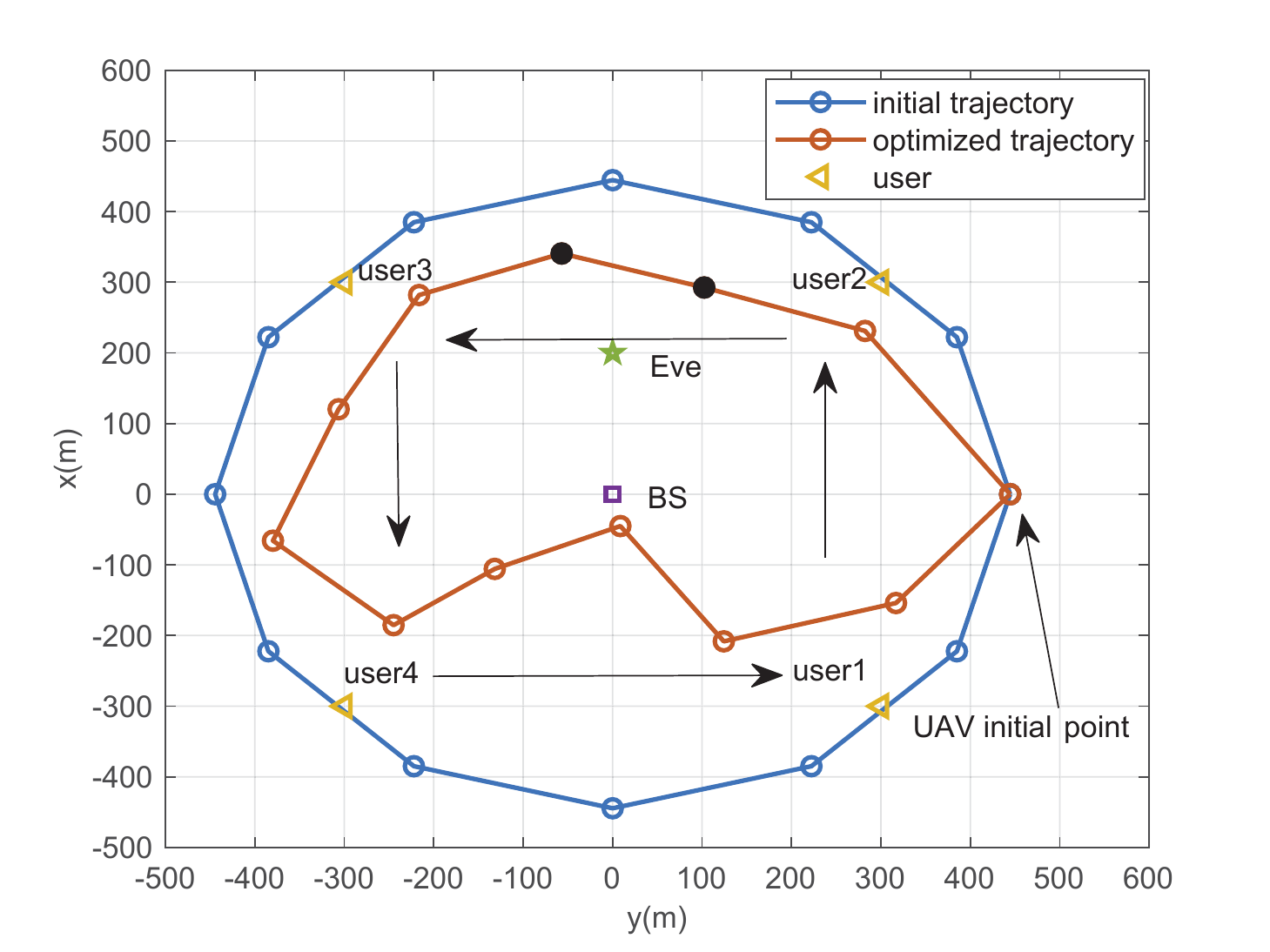}
	\caption{Optimized UAV trajectory by Algorithm 2 ($v_{\max} = 30$ m/s).} \label{figN16v30}
\end{figure}
\begin{figure}[t]
	\centering
	\includegraphics[width=3.5in]{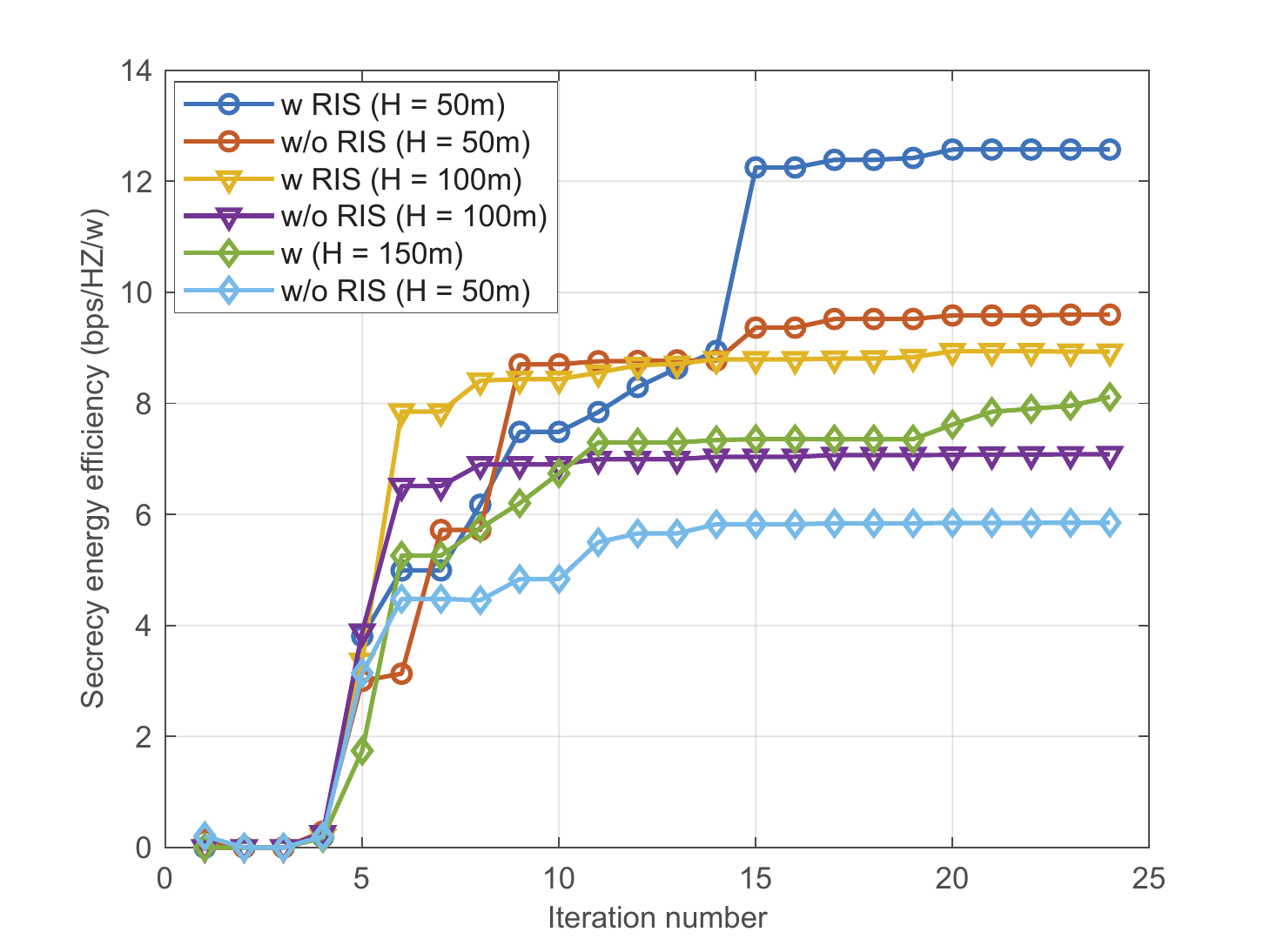}
	\caption{Convergence behavior versus the UAV's flying height $H$.} \label{figcompareH}
\end{figure}
\begin{figure}[t]
	\centering
	\includegraphics[width=3.5in]{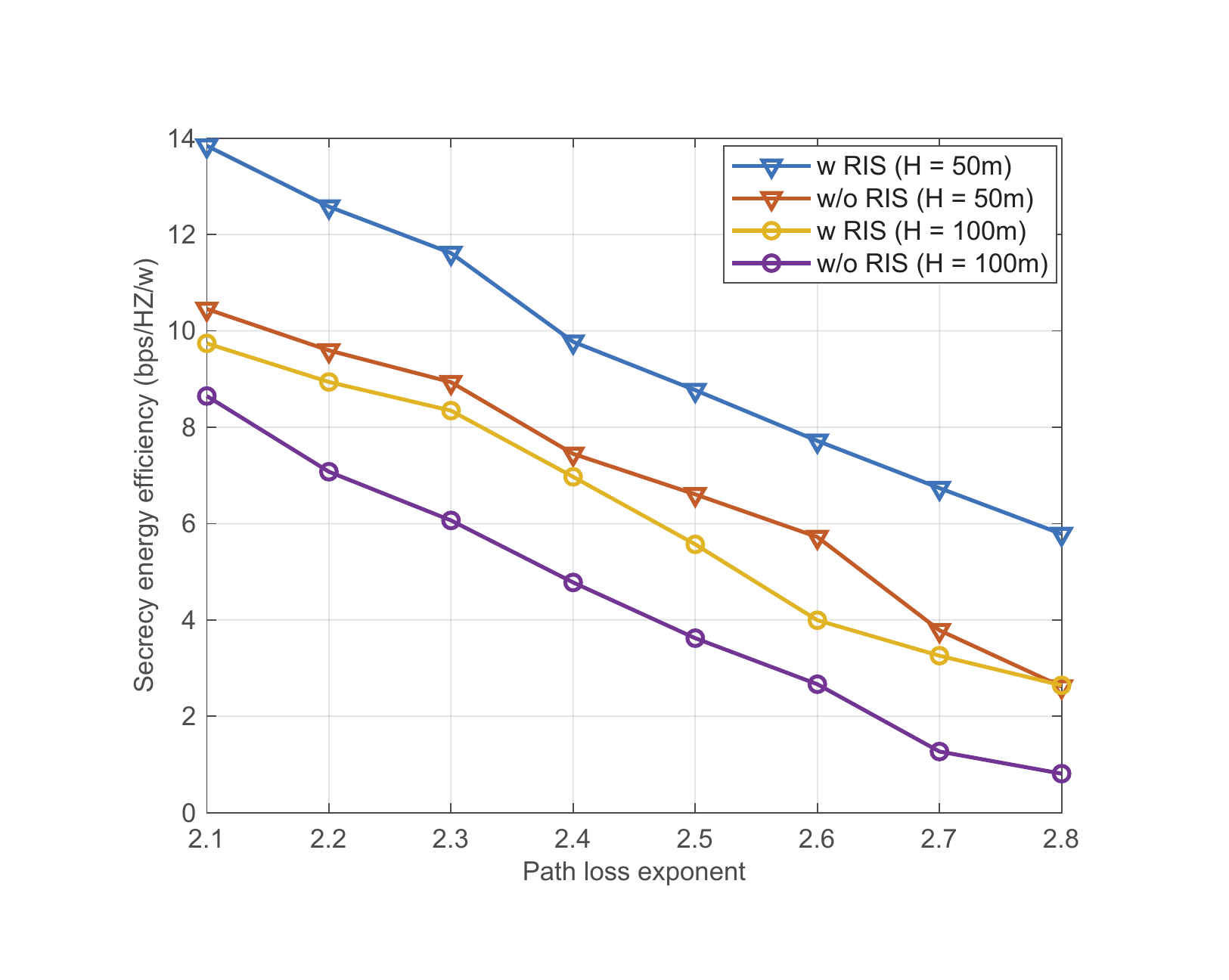}
	\vspace{-1em}
	\caption{Secrecy energy efficiency behavior versus the pathloss exponent $\alpha$.} \label{figalpha}
	\vspace{-1em}
\end{figure}
\begin{figure}[t]
	\centering
	\includegraphics[width=3.5in]{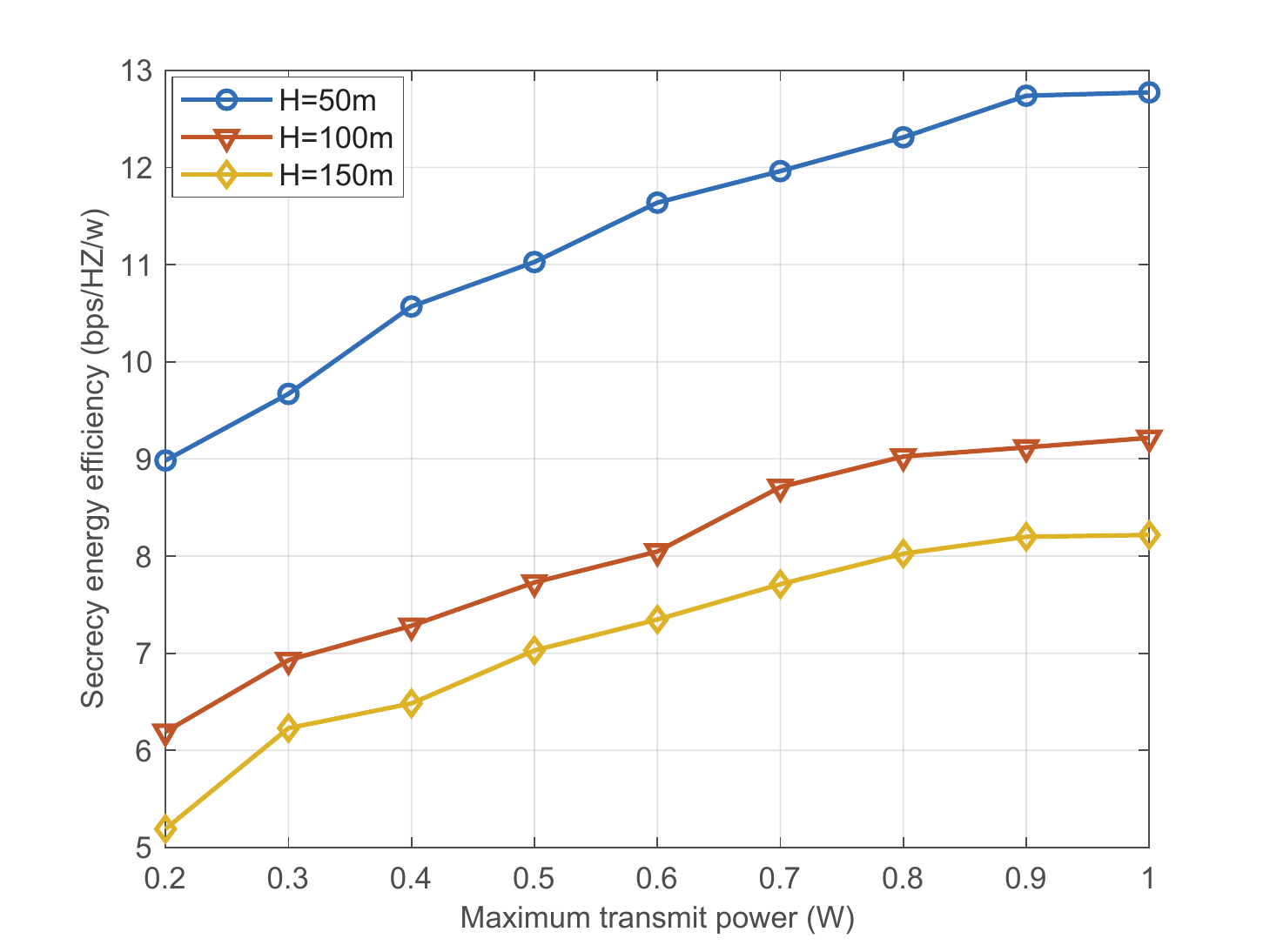}
	\vspace{-1em}
	\caption{Secrecy energy efficiency behavior versus the maximum transmit power.} \label{figp}
	\vspace{-1em}
\end{figure}
\begin{figure}[t]
	\centering
	\includegraphics[width=3.5in]{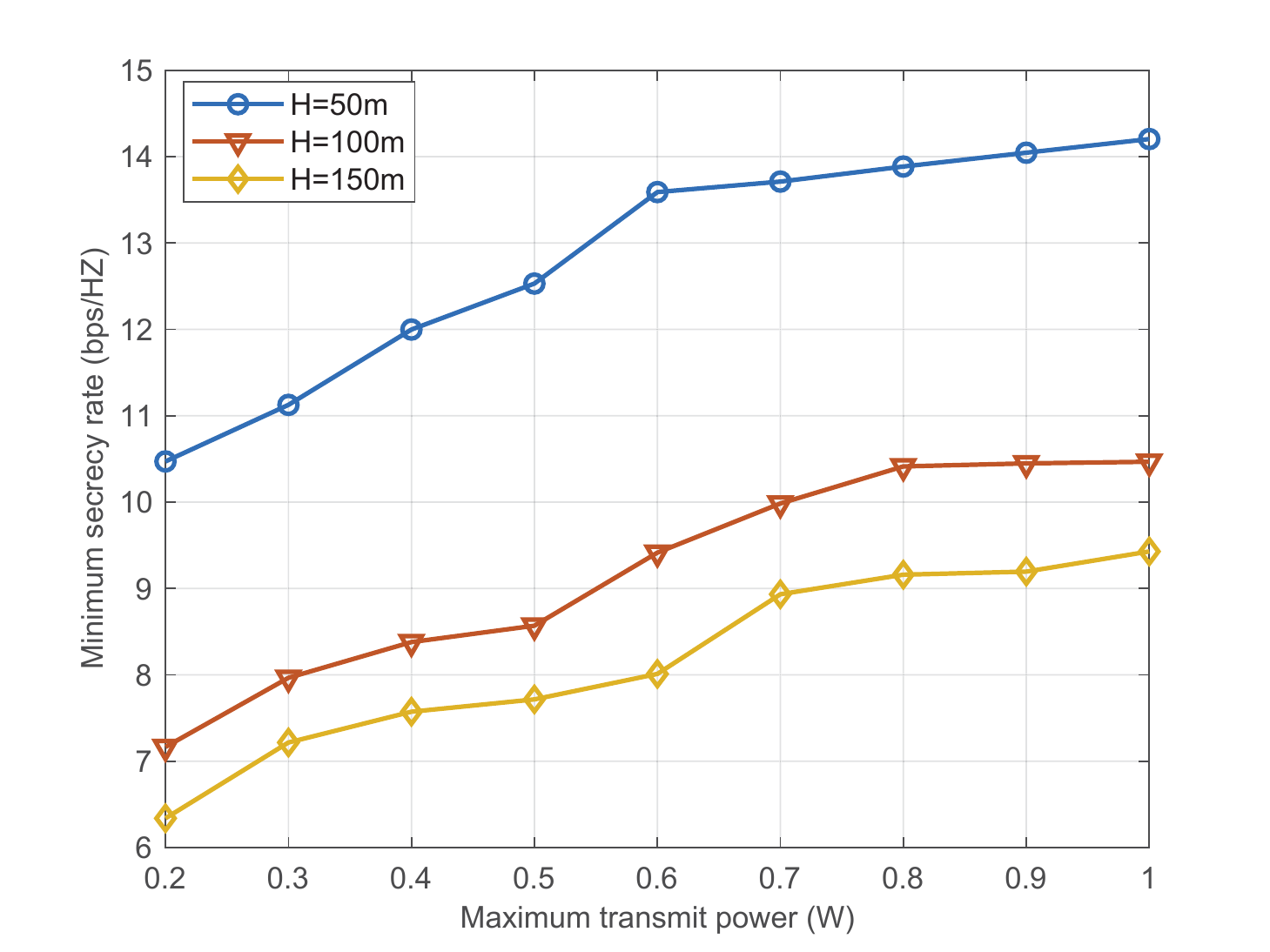}
	\vspace{-1em}
	\caption{Minimum secrecy rate behavior versus the maximum transmit power.} \label{figp2}
	\vspace{-1em}
\end{figure}

Fig. \ref{figcompareH} illustrates the convergence of the Algorithm 2 in different heights under 2 two case: with RIS (labeled 'w RIS') and the conventional AF relay
scheme without RIS  (labeled 'w/o RIS') mentioned in \cite{li2019joint}, respectively. In w/o RIS case, the UAV works as an AF relay, and its wireless transmit power $P_{UAV}$ is set to 0.2 W. It can be seen that Algorithm 2 converges fast which reveals the effectiveness of the proposed algorithm. Besides, Fig. \ref{figcompareH} also shows the relationship between the secrecy energy efficiency and the UAV flight height: the secrecy energy efficiency decreases with the increase of height. This is because with the increase of the distance of UAV-BS and UAV-users, the channel gain  decreases subsequently. Thus, the achievable secrecy rates reduce correspondingly. In addition, by comparing the RIS's performance with no-RIS in Fig. \ref{figcompareH}, the advantage of RIS in improving energy efficiency is obvious. Obviously, the UAV equipped with RIS achieves higher secrecy energy efficiency compared to the case without RIS.

The secrecy energy efficiency versus the pathloss exponent $\alpha$ is shown in Fig. \ref{figalpha}. In this figure, we can see that the energy efficiency of different UAV flying height decreases as the pathloss exponent $\alpha$ increases. This is because that both the channel gain between the UAV and BS/Eve is a decreasing function of the pathloss exponent $\alpha$. It is also demostrated that the lower flying height achieves better performance than the higher flying height. The reason is that the channel gain is a decreasing function of the distance of UAV-BS and UAV-users. When the flying height get larger, the achievable rate decreases correspondingly, as well the secrecy energy efficiency.

In Fig. \ref{figp}, we study the variation of the secure energy efficiency performance in different heights with the maximum transmit power of users (we assume the maximum transmit power $P_k$ for user $k$ are all the same) varies. It can be seen that secure energy efficiency first increases fast and then the increasing speed is slower. This is because secure energy efficiency is an non-decreasing function of the maximum transmit power $P_k$. As shown in Fig. \ref{figp}, when $P_k >$ 0.9 W ($H = 50$ m) or $P_k >$ 0.8 W ($H = 100$ m and $H = 150$ m) the secrecy energy efficiency keeps stable at a certain value. The result can be explained as follow. When maximum transmit power $P_k$ gets larger, the exceed transmit power is not used and it will not increase the energy efficiency. The result indicates that providing more transmit power for the system does not always obtain additional secure energy efficiency gains, and  reasonable designing maximum transmit power will save energy.  Fig. \ref{figp2} presents how the minimum secrecy rate changes with the increase of transmit power. It is obvious that when the transmit power becomes larger, e.g. $P_k>$ 0.7 W, the minimum secrecy rate increases slower than before. That means we do not need as much transmit power as possible to achieve the minimum secrecy rate.

In the final set of experiments shown in Fig. \ref{figM}, We compare the secrecy energy efficiency performance of the proposed scheme with RIS (labeled 'w RIS') and the traditional AF scheme without RIS (labeled 'w/o RIS'). First, Fig. \ref{figM}  shows secrecy energy efficiency in w RIS case increases as the number of reflecting elements $M$ becomes large. Second, we can see that the w RIS case performs better with the increase of reflecting elements compared to w/o RIS case. This is reasonable as RIS is an passive reflecting structure and does not cost any specific energy. In contrast, in w/o RIS case UAV needs extra energy to relay the received signal to BS which increase the energy consumption. From Fig. \ref{figM}, the proposed scheme can increase up to 49.5\% ($H = 50$ m) and 56.7\% ($H = 100$ m) compared to the traditional AF scheme, which demonstrates the RIS's characteristic in improving energy efficiency.

\begin{figure}[t]
	\centering
	\includegraphics[width=3.5in]{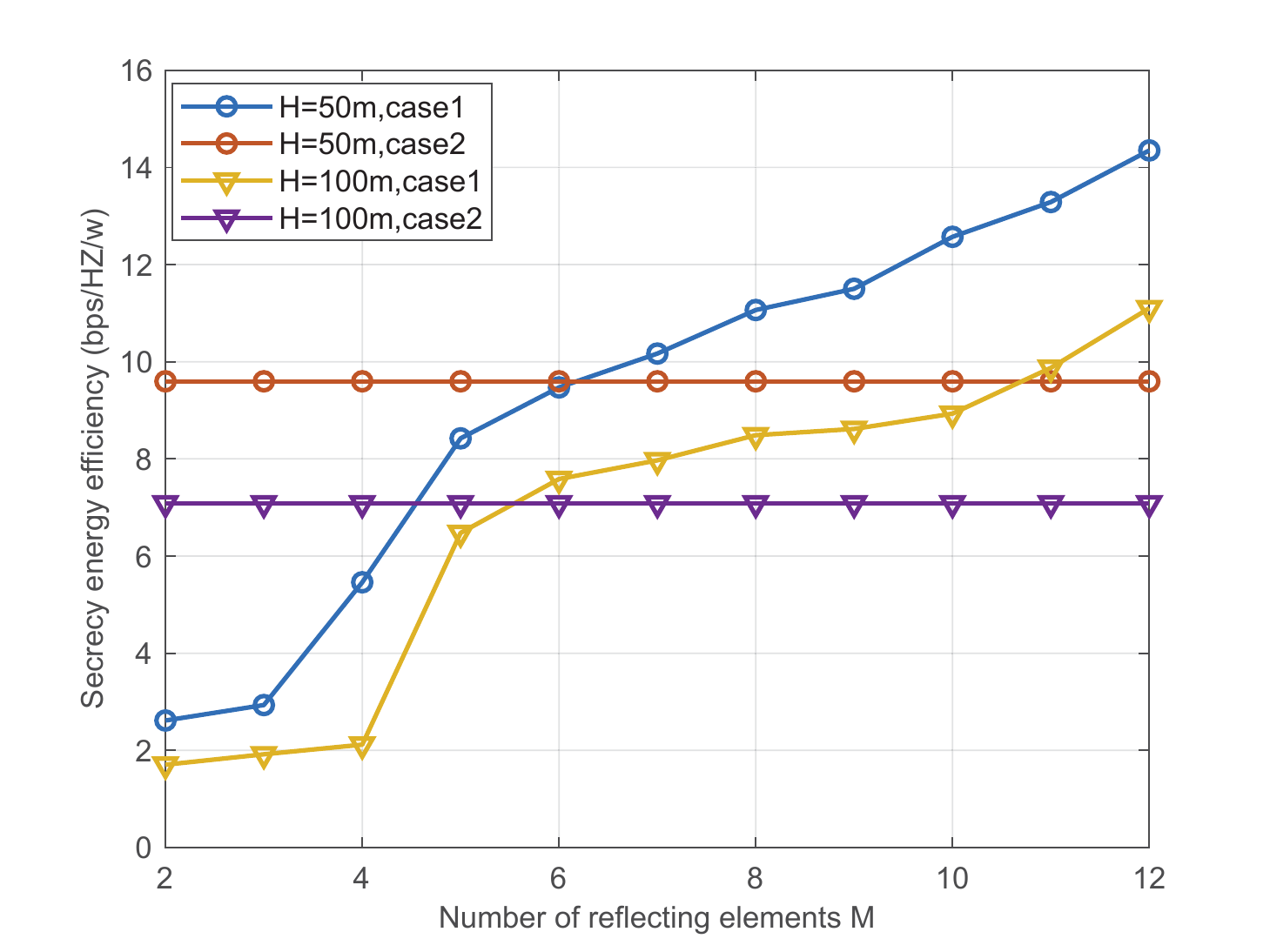}
	\caption{Secrecy energy efficiency behavior versus the number of reflecting elements $M$.} \label{figM}
\end{figure}

\section{Conclusion}
 In this paper, we have investigated an  UAV  equipped with a RIS wiretap wireless network. The RIS's phase shifts, UAV association, trajectory and the user's transmit power have been jointly optimized to maximize the system secrecy energy efficiency. To solve this problem, we have proposed an efficient iterative algorithm by applying the SCA and alternating methods. In particular, we proposed two optimization schemes to solve the joint phase and UAV trajectory optimization problem. Numerical results have shown the fast convergence of the proposed algorithm, and our proposed method provides an excellent plan for the outdoor communication scene with obstacles.

	\begin{appendices}
		\section{Proof of Lemma 1  }
		Lemma 1 is proved by the definition of convex function. First, the first-order partial derivatives of $r_k[n]$ with respect to $z_k[n]$ is given as
		\begin{align}
		\frac{\partial}{\partial z_k[n]}\left\lbrace \log_2\left(1+\frac{A}{z_k[n]}\right)\right\rbrace = \frac{-1}{\ln2}\frac{A}{\left( z_k[n]+A\right)z_k[n] }
		\end{align}
		Then, the second-order partial derivative of of $r_k[n]$ with respect to $z_k[n]$ is given as
		\begin{align}
		\frac{\partial}{\partial ^2 z_k[n]}\left\lbrace \log_2\left(1+\frac{A}{z_k[n]}\right)\right\rbrace ^2= \frac{1}{\ln2} \frac{A^2+2Az_k[n]}{z^2_k[n]\left(z_k[n]+A\right)^2}
		\end{align}
		Since $z_k[n]>0$ and $A>0$, the second-order partial derivative $\frac{\partial r_k[n]^2}{\partial ^2z_k[n]}>0$. Thus, $r_k[n]$ is a convex function of $z_k[n]$.
		\section{Proof of Lemma 2  }
		
		Lemma 2 is proved by the definition of convex function.
		As $d_{k}[n]=\sqrt{\|\boldsymbol q[n]-\boldsymbol  w_k\|^2+H^2}$, we can get
		\begin{align}
		\notag
		d_k^2[n]&=\|\boldsymbol q[n]-\boldsymbol w_k\|^2+H^2\\
		&=\left(x[n]-x_k\right)^2+\left(y[n]-y_k\right)^2+H^2.
		\end{align}
		The first-order partial derivatives of $d_k^2[n]$ with respect to $x[n]$ and $y[n]$ are given by
		\begin{equation}
		\frac{\partial d_k^2[n]}{\partial x[n]}=2\left( x[n]-x_k\right),
		\end{equation}
		\begin{equation}
		\frac{\partial d_k^2[n]}{\partial y[n]}=2\left( y[n]-y_k\right).
		\end{equation}
		So we can get the the first-order derivatives of $d_k^2[n]$ with respect to $\boldsymbol q[n]$ is given as
		\begin{align}
		\frac{\partial d_k^2[n]}{\partial \boldsymbol{q}[n]}=
		\left(
		\begin{array}{c}
		\frac{\partial d_k^2[n]}{\partial x[n]}\\
		\\
		\frac{\partial d_k^2[n]}{\partial y[n]}
		\end{array}
		\right)
		=\left(
		\begin{array}{c}
		2\left( x[n]-x_k\right) \\
		\\
		2\left( y[n]-y_k\right)
		\end{array}
		\right) = 2(\boldsymbol{q}[n]-\boldsymbol w_k).
		\end{align}
		It can be easily know the the Hessian of $d_k^2[n]$ is
		\begin{equation}
		\bigtriangledown d_k^2[n]=
		\left[
		\begin{array}{ccc}
		2 & 0\\0 & 2
		\end{array}
		\right].
		\end{equation}
		Obviously, it is positive definite. Consequently, $d_k^2[n]$ is a convex function.
		
		The first-order partial derivatives of $d_k^4[n]$ with respect to $x[n]$ and $y[n]$ are given by
		\begin{equation}
		\frac{\partial d_k^4[n]}{\partial x[n]}=2d_k^2[n]\frac{\partial d_k^2[n]}{\partial x[n]}=4d_k^2[n]\left( x[n]-x_k\right),
		\end{equation}
		\begin{equation}
		\frac{\partial d_k^4[n]}{\partial y[n]}=2d_k^2[n]\frac{\partial d_k^2[n]}{\partial y[n]}=4d_k^2[n]\left( y[n]-y_k\right).
		\end{equation}
		The second-order partial derivatives of $d_k^4[n]$ with respect to $x[n]$ and $y[n]$ are given by
		\begin{equation}
		\frac{\partial \left( d_k^4[n]\right)^2}{\partial ^2 x[n]}=4\left( d_k^2[n]+2\left( x[n]-x_k\right)^2 \right),
		\end{equation}
		\begin{equation}
		\frac{\partial \left( d_k^4[n]\right)^2}{\partial ^2 y[n]}=4\left( d_k^2[n]+2\left( y[n]-y_k\right)^2 \right),
		\end{equation}
		\begin{equation}
		\frac{\partial \left( d_k^4[n]\right)^2}{\partial  x[n]\partial  y[n]}=8\left( x[n]-x_k\right) \left( y[n]-y_k\right) = \frac{\partial \left( d_k^4[n]\right)^2}{\partial  y[n]\partial  x[n]},
		\end{equation}
		\begin{align}\label{app1}
		\notag
		&\frac{\partial \left( d_k^4[n]\right)^2}{\partial ^2 x[n]}\frac{\partial \left( d_k^4[n]\right)^2}{\partial ^2 y[n]}-\frac{\partial \left( d_k^4[n]\right)^2}{\partial  x[n]\partial  y[n]}\frac{\partial \left( d_k^4[n]\right)^2}{\partial  y[n]\partial  x[n]}\\
		=&16d_k^4[n]+32\left( x[n]-x_k\right)^2\left( y[n]-y_k\right)^2>0.
		\end{align}
		Since $\frac{\partial \left( d_k^4[n]\right)^2}{\partial ^2 x[n]}>0$ and \eqref{app1}, $d_k^4[n]$ is a convex function.
		\section{Joint Phase and UAV Trajectory Optimization (scheme II)}
		
		To reveal the hidden convexity in $\boldsymbol g_{b}^H[n]\boldsymbol\Theta[n] \boldsymbol g_k[n]$ and $\boldsymbol g_{b}^H[n]\boldsymbol\Theta[n] \boldsymbol g_k[n]$,  we first introduce two slack variables $ u_k[n]=\phi_e[n]-\phi_k[n] $, $l_k[n]=\phi_b[n]-\phi_k[n] $. To make problem more tackleable, we introduce an approximation as follow,
		\begin{equation}\label{newphik}
		\phi_k[n]=\frac{x[n]-x_k}{d_{k}[n]}\approx\frac{x[n]-x_k}{d_{k}^{(r-1)}[n]},
		\end{equation}
		\begin{equation}\label{newphib}
		\phi_b[n]=\frac{x[n]-x_b}{d_{b}[n]}\approx\frac{x[n]-x_b}{d_{b}^{(r-1)}[n]},
		\end{equation} 
		and
		\begin{equation}\label{newphie}
		\phi_e[n]=\frac{x[n]-x_e}{d_{e}[n]}\approx\frac{x[n]-x_e}{d_{e}^{(r-1)}[n]},
		\end{equation} 
		where $d_{k}^{(r-1)}[n],d_{b}^{(r-1)}[n],d_{e}^{(r-1)}[n]$ is the solution in the $(r-1)$-th iteration. It can be found that both $u_k[n]$ and $l_k[n]$ are the linear functions of $x[n]$ as well as $q[n]$ based on \eqref{newphik}-\eqref{newphie}. Note that the approximations \eqref{newphik}-\eqref{newphie} are reasonable as the distance between the users and the UAV is large and it will not change a lot in the  $(r-1)$-th iteration. To further solve the problem, we set $t_{mk}[n]=\theta_m[n]+\frac{2(m-1)\pi d}{\lambda}u_k[n]$,  $s_{mk}[n]=\theta_m[n]+\frac{2(m-1)\pi d}{\lambda}l_k[n]$. It is obvious that $t_{mk}[n]$, and $s_{mk}[n]$ are linear functions of variables $\theta_m[n]$, $u_k[n]$ and $l_k[n$], Consequently, we can obtain
		
		\begin{align}\label{eq48}	
		\notag
		&\left|\sum\limits_{m=1}^Me^{jt_{mk}[n]}\right|^2\\
		\notag
		&=\left[\sum\limits_{m=1}^Mcos(t_{mk}[n]) \right]^2 +\left[\sum\limits_{m=1}^Msin(t_{mk}[n]) \right]^2\\
		\notag
		&\approx \left[\sum\limits_{m=1}^Mcos(t_{mk}^{(r)}[n])\right]^2+\left[\sum\limits_{m=1}^Msin(t_{mk}^{(r)}[n]) \right]^2  \\
		\notag
		&-2\sum\limits_{m=1}^M\left(\sum\limits_{i=1}^Mcos(t_{ik}^{(r)}[n]) \right)sin(t_{mk}^{(r)}[n])\left( t_{mk}[n]-t_{mk}^{(r)}[n]\right)   \\
		\notag
		&+2\sum\limits_{m=1}^M\left(\sum\limits_{i=1}^Msin(t_{ik}^{(r)}[n]) \right)cos(t_{mk}^{(r)}[n])\left( t_{mk}[n]-t_{mk}^{(r)}[n]\right)\\ 
		&\triangleq g_{ek}[n]
		\end{align}
		and
		\begin{align}\label{eq49}	
		\notag
		&\left|\sum\limits_{m=1}^Me^{js_{mk}[n]}\right|^2\\
		\notag
		&=\left[\sum\limits_{m=1}^Mcos(s_{mk}[n]) \right]^2 +\left[\sum\limits_{m=1}^Msin(s_{mk}[n]) \right]^2\\
		\notag
		&\approx \left[\sum\limits_{m=1}^Mcos(s_{mk}^{(r)}[n])\right]^2+\left[\sum\limits_{m=1}^Msin(s_{mk}^{(r)}[n]) \right]^2  \\
		\notag
		&-2\sum\limits_{m=1}^M\left(\sum\limits_{i=1}^Mcos(s_{ik}^{(r)}[n]) \right)sin(s_{mk}^{(r)}[n])\left( s_{mk}[n]-s_{mk}^{(r)}[n]\right)   \\
		\notag
		&+2\sum\limits_{m=1}^M\left(\sum\limits_{i=1}^Msin(s_{ik}^{(r)}[n]) \right)cos(s_{mk}^{(r)}[n])\left( s_{mk}[n]-s_{mk}^{(r)}[n]\right)\\ 
		&\triangleq g_{bk}[n].
		\end{align}
		Note that the approximations in \eqref{eq48} and \eqref{eq49} are derived based on the first-order Taylor expansions, and both $g_{bk}[n]$  and  $g_{ek}[n]$ are linear functions with respect to $t_{mk}[n]$ and $s_{mk}[n]$.
		
		Two slack variables $\boldsymbol{G_e} = \left\{g_{ek}[n]\right\}_{\forall k,n}$,$\boldsymbol{G_b} = \left\{g_{bk}[n]\right\}_{\forall k,n}$ are introduced to assist the problem solution. Thus, we get
		\setcounter{equation}{51} 
		\begin{align}
		c_k[n] &= \log_2\left( 1+\frac{h_0^2g_{ek}[n]}{d_k^a[n]d_e^a[n]}\right) \\
		\notag
		&=\log_2\left( d_k^a[n]d_e^a[n]+h_0^2g_{ek}[n]\right) -\log_2\left( d_k^a[n]d_e^a[n]\right)
		\end{align}
		and
		\begin{align}
		r_k[n] &= \log_2\left( 1+\frac{h_0^2g_{bk}[n]}{d_k^a[n]d_b^a[n]}\right)\\
		\notag
		&=\log_2\left( d_k^a[n]d_b^a[n]+h_0^2g_{eb}[n]\right) -\log_2\left( d_k^a[n]d_b^a[n]\right),
		\end{align}
		where $g_{ek}[n]$ and $g_{eb}[n]$ both contain the variable $q[n]$. We introduce another two slack variables $\boldsymbol{Z}=\left\{z_k[n]\right\}_{\forall k,n}$ and $\boldsymbol{V} =  \left\{v_k[n]\right\}_{\forall k,n} $. 
		We reformulate the problem as
		\begin{subequations}\label{newmaxq1}
			\begin{align}
			\!\!&\mathop{\max}_{\zeta,\boldsymbol{Q},\boldsymbol{\Theta},\boldsymbol{Z},\boldsymbol{V},\boldsymbol{U},\boldsymbol{L},\boldsymbol{G_{e}},\boldsymbol{G_b}}\quad
			\zeta\\
			\textrm{s.t.}\quad\:\:
			& \zeta \leq\frac 1 N \sum_{n=1}^Na_k[n]\left[r_k[n]-c_k[n]\right]^{+},\forall k\in\mathcal K\label{newmaxq2_1}\\
			&\|\boldsymbol q[n]-\boldsymbol q[{n-1}]\| \leq S_{\max}, \quad \forall n\in\mathcal N\\
			& \boldsymbol q[N]=\boldsymbol q[0]\label{newmaxq2_2}\\
			&z^{\frac{2}{\alpha}}_k[n]\geq d_k^2[n]d_b^2[n], \quad \forall n\in\mathcal N,k\in\mathcal K\label{newmaxq2_3}\\
			&v^{\frac{2}{\alpha}}_k[n]\leq d_k^2[n]d_e^2[n], \quad \forall n\in\mathcal N,k\in\mathcal K\label{newmaxq2_4}\\
			&u_k[n]=\phi_e[n]-\phi_k[n],  \quad \forall n\in\mathcal N,k\in\mathcal K\label{newmaxq2_5}\\ 
			&l_k[n]=\phi_b[n]-\phi_k[n],  \quad \forall n\in\mathcal N,k\in\mathcal K\label{newmaxq2_6}\\
			&\theta_{m}[n] \in[0,2\pi],   \quad  \forall   m \in\mathcal M,n\in\mathcal N\label{newmaxq2_9}\\
			&\eqref{eq48},\eqref{eq49},\quad \forall n\in\mathcal N,k\in\mathcal K.
			\end{align}
		\end{subequations}
		
		To handle the nonconvexity of \eqref{newmaxq2_1}, we have
		\begin{align}
		R_k[n]&=r_k[n]-c_k[n]\\
		\notag
		& = \log_2\left( z_k[n]+h_0^2g_{bk}[n]\right) -\log_2\left( z_k[n]\right)\\
		\notag
		&-\log_2\left( v_k[n]+h_0^2g_{ek}[n]\right) +\log_2\left( v_k[n]\right)\\
		\notag
		&\geq \log_2\left( z_k[n]+h_0^2g_{bk}[n]\right)+\log_2\left( v_k[n]\right)\\
		\notag
		&-\log_2\left( z_k^{(l)}[n]\right)-\frac{1}{\ln2 z_k^{(l)}[n]}\left(z_k[n]-z_k^{(l)}[n] \right)\\
		\notag
		&-\log_2\left( v_k^{(l)}[n]+h_0^2g_{ek}^{(l)}[n]\right)\\
		\notag
		&-\frac{1}{(\ln2)(v_k^{(l)}[n]+h_0^2g_{ek}^{(l)}[n])}\left( v_k[n]-v_k^{(l)}[n]\right)\\
		\notag
		&-\frac{h_0^2}{(\ln2)(v_k^{(l)}[n]+h_0^2g_{ek}^{(l)}[n])}\left(g_{ek}[n]-g_{ek}^{(l)}[n]\right)\\
		&\triangleq \hat{R}_k[n].
		\end{align}	
	
	To handle the nonconvexity of \eqref{newmaxq2_3} and \eqref{newmaxq2_4}, the technique applied is the same as the \eqref{maxq7}-\eqref{maxq9}.
		Hence, we can solve the following problem
		\begin{subequations}\label{newmaxqall1_1}
			\begin{align}
			\!\!&\mathop{\max}_{\zeta,\boldsymbol{Q},\boldsymbol{\Theta},\boldsymbol{Z},\boldsymbol{V},\boldsymbol{U},\boldsymbol{L},\boldsymbol{G_e},\boldsymbol{G_b}}\quad
			\zeta\\
			\textrm{s.t.}\quad\:\:
			& \zeta \leq\frac 1 N \sum_{n=1}^Na_k[n]\max\left\{\hat{R}_k[n],0\right\},\forall k\in\mathcal K\label{newmaxq1_1}\\
			&\|\boldsymbol q[n]-\boldsymbol q[{n-1}]\| \leq S_{\max}, \quad \forall n\in\mathcal N\label{newmaxq1_2}\\
			& \boldsymbol q[N]=\boldsymbol q[0]\label{newmaxq1_3}\\
			&z^{\frac{2}{\alpha}}_k[n]\geq f\left(\boldsymbol{q}[n]\right)\quad \forall n\in\mathcal N,k\in\mathcal K\label{newmaxq1_4}\\
			&h\left(v_k[n]\right)\leq g\left(\boldsymbol{q}[n]\right)\quad \forall n\in\mathcal N,k\in\mathcal K\label{newmaxq1_5}\\
			&u_k[n]=\phi_e[n]-\phi_k[n],  \forall n\in\mathcal N,k\in\mathcal K\label{newmaxq1_6}\\
			&l_k[n]=\phi_b[n]-\phi_k[n],  \forall n\in\mathcal N,k\in\mathcal K \label{newmaxq1_7}\\
			&\theta_{m}[n] \in[0,2\pi],   \quad  \forall   m \in\mathcal M,n\in\mathcal N\label{newmaxq1_10}\\
			&\eqref{eq48},\eqref{eq49},\quad \forall n\in\mathcal N,k\in\mathcal K.
			\end{align}
		\end{subequations}
	Since \eqref{newmaxq1_6}, \eqref{newmaxq1_7} and \eqref{newmaxq1_10} are all linear constraints and other constrains are convex, problem \eqref{newmaxqall1_1} is a convex optimization problem which can be efficiency solved by standard convex optimization solver. As too many approximation are applied in the proposed optimization scheme II, the optimal objective value obtained from problem \eqref{newmaxqall1_1} can serve as a lower bound of problem \eqref{maxq}.
		
	\section{Proof of Lemma 3  }
	
The proof is eastablished by showing that the secrecy energy efficiency is non-decreasing after each iteration. Considering in the $t$-th iteration, we proved following update rules based on Algorithm 2
	\begin{align}\label{convergence}
	\notag
	\Gamma^{(t)}&=f(\boldsymbol{A}^{(t)},\boldsymbol{P}^{(t)},\boldsymbol{Q}^{(t)},\boldsymbol{\Theta}^{(t)})\\
	\notag
	&\overset{(a)}{\leq}f(\boldsymbol{A}^{(t+1)},\boldsymbol{P}^{(t)},\boldsymbol{Q}^{(t)},\boldsymbol{\Theta}^{(t)})\\
	\notag
	&\overset{(b)}{\leq}f(\boldsymbol{A}^{(t+1)},\boldsymbol{P}^{(t+1)},\boldsymbol{Q}^{(t)},\boldsymbol{\Theta}^{(t)})\\
	\notag
	&\overset{(c)}{\leq}f(\boldsymbol{A}^{(t+1)},\boldsymbol{P}^{(t+1)},\boldsymbol{Q}^{(t+1)},\boldsymbol{\Theta}^{(t+1)})\\
	&=\Gamma^{(t+1)}
	\end{align}
where $\Gamma=f(\boldsymbol{A},\boldsymbol{P},\boldsymbol{Q},\boldsymbol{\Theta})=\frac{\zeta}{\sum_{k=1}^K \sum_{n=1}^N p_k[n]+P_0}$. Inequality $(a)$ follows that in step 3 in Algorithm 2 problem \eqref{maxa} is solved optimally with solution $\boldsymbol{A}^{(t+1)}$. In the proposed Algorithm 2, for the power control optimization problem \eqref{maxp} and trajectory and phift optimization \eqref{maxq}, we only solved the their approximate problem \eqref{maxpnew2}  and problem \eqref{maxq8}. Define $\Gamma_{pow}=f_{pow}(\boldsymbol{A},\boldsymbol{P},\boldsymbol{Q},\boldsymbol{\Theta})$ , where $\Gamma_{pow}$ is respectively the objective values of problem \eqref{maxpnew2}. The inequality $(b)$ can be explained as follow
\begin{small}
\begin{align}
\notag
f(\boldsymbol{A}^{(t+1)},\boldsymbol{P}^{(t)},\boldsymbol{Q}^{(t)},\boldsymbol{\Theta}^{(t)})&\overset{(d)}= f_{pow}(\boldsymbol{A}^{(t+1)},\boldsymbol{P}^{(t)},\boldsymbol{Q}^{(t)},\boldsymbol{\Theta}^{(t)})\\
\notag
&\overset{(e)}{\leq} f_{pow}^{lb}(\boldsymbol{A}^{(t+1)},\boldsymbol{P}^{(t+1)},\boldsymbol{Q}^{(t)},\boldsymbol{\Theta}^{(t)})\\
&\overset{(f)}{\leq}f(\boldsymbol{A}^{(t+1)},\boldsymbol{P}^{(t+1)},\boldsymbol{Q}^{(t)},\boldsymbol{\Theta}^{(t)})
\end{align}
\end{small}
where inequality $(d)$ holds since the SCA method in \eqref{sca}, which means problem \eqref{maxpnew2} at the $p_k^{(r)}$ has the same objective value as problem \eqref{maxp}; inequality $(e)$ follows that the step 4 of Algorithm 2 with given $\boldsymbol{A}^{(t+1)}$, $\boldsymbol{Q}^{(t)}$ and $\boldsymbol{\Theta}^{(t)}$, $\boldsymbol{P}^{(t+1)}$ is the optimal solution of problem \eqref{maxpnew2};
inequality (f) is due to the fact that the objective vaule of problem \eqref{maxq8} is the lower bound of the original problem \eqref{maxq} at $\boldsymbol{P}^{(t+1)}$. The inequality $(c)$ can hold due to the fact that we set $Q^{(t+1)}=Q^{(t)}, \Theta^{(t+1)}=\Theta^{(t)}$ if $\Gamma^{(t+1)}< \Gamma^{(t)}$ in Algorithm 2.
Thus the object vaule is non-decreasing after each iteration of Algorithm 2. Furthermore, the objective vaule of problem \eqref{max1} is upper bounded by a finite value, the proposed Algorithm is guaranteed to converge.
	
\end{appendices}

\section*{Acknowledgments}
This work was supported by the National Natural Science Foundation of China (NSFC) under grant 61871128 and the Fundamental Research on Foreword Leading Technology of Jiangsu Province under grant BK20192002.

	\bibliographystyle{IEEEtran}
	\bibliography{IEEEabrv,MMM}

\end{document}